\documentclass[12pt]{article}
\usepackage{setspace}
\usepackage{ulem}
\usepackage{amsmath,amsxtra,latexsym,amsthm,amssymb,amscd}
\usepackage[english]{babel}
\usepackage{subfig}
\usepackage{graphicx}
\usepackage{enumerate} 
\usepackage{geometry}
\usepackage{mathrsfs}
\usepackage{graphics}
\usepackage{epsfig}
\usepackage{epic}
\usepackage{footmisc}
\usepackage{amsfonts}
\usepackage{scrextend}
\usepackage[colorinlistoftodos]{todonotes}
\usepackage[metapost]{mfpic}
\usepackage{tabularx}
\usepackage{relsize}
\usepackage{graphicx,natbib} 
\usepackage{blindtext}
\usepackage[colorlinks = true,
            linkcolor = blue,
            urlcolor  = blue,
            citecolor = blue,
            anchorcolor = blue]{hyperref} 
            
\usepackage{flexisym}
\newtheorem{theorem}{Theorem}
\newtheorem{proposition}{Proposition}

\newtheorem{lemma}{Lemma}

\theoremstyle{definition}
\newtheorem{remark}{Remark}
\newtheorem{axiom}{Axiom}
\newtheorem{definition}{Definition}
\newtheorem{example}{Example}
\newcommand{\co}{\mathrm{co}}

\newtheorem*{yhtheorem}{Yosida-Hewitt decomposition theorem}

\newtheorem*{linva}{Invariance with respect to applying $T$ to improving sequences (I$T$IS)}
\newtheorem*{linvashift}{Invariance with respect to delaying improving sequences (IDIS)}
\newtheorem*{patience}{Patience}
\newtheorem*{part1}{Only-if part}
\newtheorem*{part2}{If part}

\newtheorem*{fpermuinva}{Invariance with respect to finite permutation of improving sequences (IFPIS)}
\newtheorem*{permuinva}{Invariance with respect to permutation of improving sequences (IPIS)}

\newtheorem*{Schau}{Schauder–Tychonoff fixed-point theorem}
\newtheorem*{bicon}{Biconjugation theorem}
\newtheorem*{timeinva}{Time invariance}
\newtheorem*{pareto}{Unanimity}

\DeclareMathOperator*{\argmin}{arg\,min}

\begin{document}

\title{Dynamic choices, temporal invariance and variational discounting\thanks{We thank Federico Echenique, Xiangyu Qu, Jean-Marc Tallon, Eric Danan, Marcus Pivato, Chew Soo Hong, Lorenzo Stanca and participants at many seminars and conferences for their valuable comments.}} 

\author{Bach Dong-Xuan\thanks{Corresponding author. Center for Mathematical Economics, Bielefeld University. E-mail: \url{bachdongx@gmail.com}.} \and Philippe Bich\thanks{Universit\'e Paris 1 Panth\'eon-Sorbonne \& Paris School of Economics. E-mail: \url{bich@univ-paris1.fr}.}} 

\date{June 17, 2024}


\maketitle

\begin{abstract}

People often face trade-offs between costs and benefits occurring at various points in time. The predominant discounting approach is to use the exponential form. Central to this approach is the discount rate, a unique parameter that converts a future value into its present equivalent. However, a universally accepted discount rate remains a matter of ongoing debate and lacks consensus. This paper provides a robust solution for resolving conflicts in discount rates, which recommends considering all discount rates but aims to assign varying degrees of importance to these rates. Moreover, a considerable number of economists support a theory that suggests equal consideration of future and present utilities. In response to this debate, we introduce a general criterion capable of accommodating situations where it is feasible not to discount future utilities. This criterion encompasses and extends various existing criteria in the literature.  
\end{abstract}



\newpage

\section{Introduction}

A government agency frequently confronts policies or projects that have significant and long-term consequences, such as initiatives aimed at reducing carbon emissions and fostering the adoption of green technology. The workhorse model used to evaluate the feasibility of these decisions is the Exponential Discounting Utility (EDU) model introduced by \cite{R1928} and \cite{S1937}. At the core of this model is the concept of the discount rate, which plays a central role in converting a value received in a future time period to an equivalent value received immediately.


Yet, determining the most suitable discount rate remains a contentious and intricate endeavor, mainly because there is no universally agreed-upon discount rate. The diverging viewpoints on discount rates largely emanate from normative disagreements regarding the societal time preferences that dictate marginal rates of substitution, as observed in the works of \cite{DFMG2018} and \cite{D2008}. Additionally, the choice of the discount rate can significantly impact the policies and decision outcomes. Even small adjustments in the discount rate can produce a profound impact on the estimated value of public projects, for example, infrastructure investments, climate change mitigation strategies, and nuclear waste management (\cite{article}). This influence extends to important areas like fiscal policy (\cite{barro}) and environmental policy (\cite{nordhaus}). Illustrating the significance of this matter, \cite{weitzman} argues that determining the discount rate for evaluating long-term public projects is ``one of the most critical problems in all of economics."


Solutions for resolving conflicts in discount rates discussed in the existing literature primarily center around utilitarian and maxmin approaches. The utilitarian approach acknowledges the decision maker's discretion in assigning varying degrees of importance to different discount rates. However, a prevailing argument suggests that crucial decision-making situations demand a certain level of caution (see, for example, \cite{AG2016}), which the utilitarian approach does not adequately address. On the other hand, the maxmin approach can accommodate cautious decision-making, yet it treats all ``corrected" discount rates equally.

In this paper, we propose a ``robust" solution to confront the challenge of dealing with multiple discount rates. This approach takes into account both the need for caution in decision-making and the variability in judgments across discount rates. Specifically, we introduce and characterize the following criterion: 
\begin{equation}\label{eqvar}
    I(x)=\min_{\delta \in [0,1)} \bigg\{(1-\delta) \sum_{t=0}^{+\infty} \delta^t  x_t + c(\delta)\bigg\}.
\end{equation}
This criterion aligns with the structure of a variational representation introduced by \cite{M2006} in the context of decision-making under uncertainty. Thus, we refer to this criterion as the \textit{variational discounting} criterion. According to this criterion, the decision maker considers all discount rates when evaluating long-term projects but wishes to give some of them significantly greater weight, and these weights are expressed through the ``cost" function $c$.

The foundation of the variational discounting representation rests on a new invariance principle called \textit{Invariance with respect to delaying improving sequences} (IDIS). In the context of this principle, an \textit{improving sequence} can be seen as a utility stream generated by a project or policy that is profitable. IDIS implies that if a project generates benefits, it remains profitable even if its implementation is delayed by one period. To the best of our knowledge, this type of invariance principle is novel in the literature.

To make well-informed intertemporal decisions, decision makers often seek guidance on determining the appropriate discount rate from a panel of experts. It is not surprising that each expert may advocate for a different discount rate. As an illustration, in a survey conducted by \cite{weitzman} on the discount rates employed by many economists for the evaluation of long-term projects, the average response is 3.96\%, with a standard deviation of 2.94\%. We show that when a variational discounting decision maker follows the standard Pareto assumption, then she only considers the discount rates recommended by these experts. In this case, the cost function $c$ can be interpreted as a measure of the decision maker's confidence in these experts.



Also, we extend representation (\ref{eqvar}) into the following criterion, which encompasses and generalizes several existing criteria in the literature:
\begin{equation}\label{eqgen}
I(x)=\min_{p \in \Delta (T^*)}\{\langle x,p \rangle + c(p)\},
\end{equation} where $\Delta (T^*)$ is the set of weights or discount structures that are eigenvectors of some linear operator $T^*$ from the space of finitely additive measures to itself. Specifically, the worst-case scenario in the prudent representation is taken over discount structures that are ``invariant" to the linear transformation $T^*$ in the sense that applying $T^*$ to them does not change their direction (it only scales them by a positive scalar).

As criterion (\ref{eqgen}) does not require \textit{discount structures} to be countably additive, we can cover situations where the decision maker finds discount structures that do not discount future utilities to be plausible. The theory of not discounting future utilities has received substantial support in the past. Typically, time periods can be viewed as representing different generations. The EDU criterion implies that the utility experienced today carries considerably more weight than the same utility experienced in the distant future. Thus, it goes against policies that prioritize the well-being of future generations. This, however, contradicts the observation that many governments are willing to invest resources to prevent climate change, which demonstrates that we are seriously care about the long-run future. Moreover, from a normative point of view, all generations should be deemed equally (see, for example, \cite{S1874} and \cite{R1928}). 

The reason why the variational discounting criterion fails to capture equal treatment of time periods remains in the \textit{Monotone Continuity}\footnote{See Subsection \ref{section3.1} for a formal definition of discussion.} assumption. As a result of discarding the assumption, we obtain the following criterion:
\begin{equation*}
I(x)=\min_{p \in \mathcal{D}}\{\langle x,p \rangle + c(p)\},
\end{equation*} where $\mathcal{D}$ is the set of all exponential discounting sequences and \textit{Banach-Mazur limits}. Notice that Banach-Mazur limits assess any utility stream $(x_0, x_1, x_2, \dots)$ and its shifted version $(x_1, x_2, x_3, \dots)$ equally, indicating a characteristic of patience. This criterion suggests that, besides considering exponential discounting sequences, the decision maker also takes into account patient discount structures.

In addition, we characterize the preferences that exhibit \textit{Invariance with respect to finite permutation of improving sequences}. More precisely, for any improving sequence, the utility stream obtained by applying any finite permutation to the improving sequence would still be considered acceptable. This captures patience in the context of improving sequences (each date is treated equally when evaluating the changes in utility resulting from a new project). We show that the set of discount structures in our representation includes only \textit{purely finitely additive} discount structures.



A key technical contribution in our paper lies in the proof of our main result, which combines functional analysis techniques with the Schauder–Tychonoff fixed-point theorem. To gain insight into the role of fixed points within our proof, consider that a discounting sequence $p$ can be characterized as an exponential discounting if and only if it is either equal to $(1, 0, 0, ...)$ or if it represents a fixed point of the transformation defined by $s(p) = \frac{1}{1-p_0}(p_1, p_2, ...)$. The existence of such fixed points is established through the Schauder–Tychonoff fixed-point theorem, a fundamental result in topology and functional analysis.\footnote{To the best of our knowledge, this is the first time in intertemporal social choice theory that a representation theorem is obtained by mixing eigenvalues, eigenvectors and a fixed-point theorem.}

The paper's structure can be summarized as follows: In Section \ref{section2}, we introduce the setting and our primary representation. Section \ref{section3} outlines a set of axioms and presents our main representation theorem. Section \ref{section4} extends to a more comprehensive result, encompassing variational representations that can incorporate discount structures without discounting future utilities (refer to Section \ref{sec:equal}). For the sake of completeness, all proofs are provided in the Appendix.

\subsection{Related literature}



In a related body of literature, researchers have explored social choice theory approaches for aggregating time preferences and reconciling diverse opinions on discount rates. A notable work in this domain is \cite{weitzman}, who conducts a survey of economists' discount rates to underpin a gamma discounting model. It is widely acknowledged that the simultaneous satisfaction of the assumption that individuals follow the EDU model and the Pareto condition is unattainable, as documented in various studies (see, for example, \cite{GZ2005}, \cite{zuber}, \cite{jackson}, and \cite{billotime}). However, \cite{MH2018} demonstrate that this negative result can be circumvented if we only require that the planner is time-consistent.\footnote{The decision maker adhering to the variational discounting model exhibits time inconsistency. The literature explores various reasons for time inconsistency in the objective functions of policymakers and social planners. A detailed review of these arguments is provided by \cite{H2020}. For policymakers, political turnovers between different parties can lead to time inconsistency (e.g., \cite{PS1989} and \cite{AG1990}). Similarly, for benevolent social planners aggregating individual preferences, factors like agents' altruism or varying discount rates can also induce time inconsistency (e.g., \cite{jackson} and \cite{GS2017}).} In their study, \cite{FK2018} introduced a modified Pareto condition that can be compatible with the assumption that individuals follow the EDU model.

\cite{CE2018} propose three rules for aggregating discount rates. One of these rules combines exponential individual discount functions in a utilitarian manner, while the other two rules aggregate exponential individual discount functions by selecting, respectively, the most pessimistic utilitarian weight and discount function. 

Recently, \cite{millner} highlights that planners can hold differing views on all the normative aspects that typically lead to debates on social discounting. Despite these differences, Millner's work shows that such disagreements can still lead to a consensus regarding the long-run discount rate.

Equal treatment for future generations gained prominence in mainstream economics through Pigou's early work on welfare (\cite{P1912}). This notion is often encapsulated by \textit{Patience}, which dictates that preferences remain intact when the utility levels of a finite number of generations are permuted. However, Patience is incompatible with the strong Pareto axiom, stating that preferences are influenced by an increase in the utility level of any generation (see, for example, \cite{D1965}, \cite{BM2003}, \cite{FM2003}, and \cite{Z2007}).

The first modern mathematical formulation of long-run intergenerational welfare was presented by \cite{R1928}. He develops a theory that does not discount later enjoyments in comparison with earlier ones. Subsequently, various criteria have been proposed for long-run intergenerational welfare. \cite{W1965} introduces the concept of overtaking optimality, while \cite{R2007} suggests the use of a purely finitely additive discount structure to discount utilities at different time points. \cite{M1998} characterizes the \textit{liminf} criterion and the \textit{Polya} index, representing agents with ``infinite" patience. More recently, \cite{P2022} establishes axiomatic foundations for the \textit{Cesàro average utility} representation. Notably, \cite{KS2018} demonstrate that many long-run social welfare functions previously studied in the literature can be expressed in a maxmin form.




\section{Setting and representation}\label{section2}

\subsection{Setting}\label{subsection2.1}

We consider a decision maker making choices among infinitely countable streams of utilities. The set $\mathbb{N}$ is the set of time periods (0 being the present). Formally, the preferences of the decision maker are represented by a binary relation $\succsim$ on $l_{\infty}$, the space of real-valued bounded sequences, endowed with the sup norm.\footnote{It is important to note that these utility flows are fundamentally distinct from financial flows. More precisely, the theory studies how individuals discount the satisfaction and discomfort experienced at specific time points, rather than how they perceive the timing of financial transactions.}  Let us recall that $l_{\infty}$ can be endowed with the following standard order relations: for every $x,y \in l_{\infty}$, $x \geq y$ if and only if  $x_n \geq y_n$ for all $n \in \mathbb{N}$, $x>y$ if and only if $x\geq y$ and $x \neq y$, and $x\gg y$ if and only if $x_n > y_n$ for all $n \in \mathbb{N}$. For every $x=(x_n)_{n \in \mathbb{N}} \in  l_{\infty}$, $x_n$ represents the utility that the decision maker obtains at time $n$.


The preference relation of the decision maker $\succsim$ is said to be a \textit{weak order} if it is \textit{complete} (i.e., for every $x,y \in l_{\infty}$, either $x\succsim y$ or $y \succsim x$) and \textit{transitive} (i.e., for every $x,y,z \in l_{\infty}$, $x\succsim y$ and $y\succsim z$ implies $x\succsim z$). We define the asymmetric part $\succ$ and symmetric part $\sim$ of $\succsim$ as usual. That is, $x \succ y$ if $x\succsim y$ and not $y\succsim x$, and $x \sim y$ if $x\succsim y$ and $y \succsim x$. For any real number $\theta$, we employ a slight abuse of notation by using $\theta$ to represent the constant sequence $(\theta, \theta, \dots)$. Additionally, for any $x \in l_{\infty}$, we refer to $I(x) \in \mathbb{R}$ as a \textit{constant equivalent} of $x$ if $x \sim I(x)$. 

If $(E_n)_{n\in \mathbb{N}}$ denotes a sequence of subsets of $\mathbb{N}$, then $(E_n)\downarrow \emptyset$ means that $E_n\supseteq E_{n+1}$ for every $n \geq 0$ and $\bigcap_{n=0}^{\infty} E_n = \emptyset$. For any $x,y \in l_{\infty}$ and $A\subseteq \mathbb{N}$, the notation $xAy$ denotes the sequence $z$ defined as follows: $z_n = x_n$ if $n \in A$, and $z_n = y_n$ if $n \in A^{c}$.

The norm dual of the space of bounded sequences $(l_{\infty},\|.\|_{\infty})$, i.e., the set of continuous linear functions from $l_{\infty}$ to $\mathbb{R}$, is denoted by $ba(\mathbb{N})$. If $\mu \in ba(\mathbb{N})$ and $x \in l_{\infty}$, we will denote (as usual) $\langle x,\mu\rangle$, rather than $\mu(x)$, the value of $\mu$ at $x$. By a Riesz representation theorem, this space can also be seen as the space of all finitely additive measures of bounded variation on the $\sigma$-algebra $2^{\mathbb{N}}$ of all subsets of $\mathbb{N}$. We endow, as usual, $ba(\mathbb{N})$ with the weak$^{\star}$ topology $\sigma(ba(\mathbb{N}),l_{\infty})$.\footnote{A net $(\mu_d)$ in $ba(\mathbb{N})$ converges to $\mu$ for this topology if and only if, for every $x \in l_{\infty}$, the net $\langle x, \mu_{d}\rangle$ converges to $\langle x, \mu \rangle$.} When $\mu$ is countably additive, we can represent $\mu$ by $(\mu_n)_{n \in \mathbb{N}}$, where $\mu_n=\mu(\{n\})$. In that case, we can write $\langle x,\mu\rangle=\int_{\mathbb{N}}x_n d\mu=\sum_{n=0}^{+\infty} x_n \mu_n$. Let $ca(\mathbb{N})$ be the set of all countably additive set functions. 

In the context of temporal decision-making, a key ingredient is how individuals assign value to different points in time. A \textit{discount structure} is formally defined as a measure $\mu \in ba(\mathbb{N})$ that satisfy the following conditions:
\begin{enumerate}[(i)]
\item $\mu (\emptyset) = 0$.
\item $\mu (A \cup B) = \mu (A) + \mu (B)$ whenever $A$ and $B$ are disjoint.
\item $\mu (\mathbb{N})=1$. 
\end{enumerate}
We denote the set of all discount structures as $\Delta$ and the set of all countably additive discount structures as $\Delta^{\sigma}$. Unless otherwise specified, any subset of $ba(\mathbb{N})$ (including $\Delta$ and $\Delta^{\sigma}$) is equipped with the relative topology derived from the weak$^{\star}$ topology on $ba(\mathbb{N})$.

\subsection{Representation}


As in \cite{M2006}, a function $c:X \rightarrow [0,\infty]$ is said to be \textit{grounded} if $\inf_{p \in X} c(p) = 0$. 

\begin{definition}
A binary relation $\succsim$ on $l_{\infty}$ admits a \textit{variational discounting} representation if there exists a grounded and lower semicontinuous function $c:[0,1] \to [0,\infty]$, where $c(1) = \infty$, such that the constant equivalent representing $\succsim$ can be expressed as
$$I(x)=\min_{\delta \in [0,1)} \bigg\{(1-\delta) \sum_{t=0}^{+\infty} \delta^t  x_t + c(\delta)\bigg\}.$$
\end{definition}

Unfortunately, a unique characterization of the cost function $c$ is not attainable. However, it can be shown that the cost function is unique up to a lower semicontinuous convex extension. To elaborate, if we view each discount factor $\delta$ as a geometric distribution $(1-\delta)(1, \delta, \delta^2, \dots)$, then $c$ can be perceived as a function whose domain is a subset of $\Delta$. Let $D$ be the closed convex hull of the set $F:= \{(1-\delta)(1,\delta,\delta^2,\dots): \delta\in [0,1)\}$. Let $c_1$ and $c_2$ be two cost functions that represent the same variational discounting preference. Clearly, for $i\in \{0,1\}$, we can express $$I(x)=\min_{p\in D} \{ \langle x,p\rangle + \bar{c}_i(p)\},$$ where $\bar{c}_i: D \to [0,\infty]$ is defined as $\bar{c}_i ((1-\delta)(1,\delta,\delta^2,\dots)) = c_i(\delta)$ for every $\delta \in [0,1)$, and $\bar{c}_i (p) = \infty$ for every $p \in D \setminus F$. We can prove that the lower semicontinuous convex envelopes of $\bar{c}_1$ and $\bar{c}_2$ coincide (for a detailed proof, please refer to Appendix \ref{appenunique}).

The decision maker who follows the variational discounting criterion simultaneously considers all discount factors to evaluate utility streams, but intends to assign varying degrees of importance or weight to these factors. The weights assigned to discount factors are measured by the ``cost" function $c$. The decision maker then conducts an evaluation of any intertemporal utility stream $x$ by identifying the minimum value of its discounted sum across all discount factors $\delta$, augmented by their corresponding costs $c(\delta)$. Discount factors characterized by higher costs receive relatively less weight. As an example, when the cost function $c(\delta)$ equals zero, the decision maker is fully confident that $\delta$ is the correct and appropriate discount factor. Conversely, when $c(\delta)$ equals infinity, this signifies that $\delta$ is deemed maximally implausible, and the decision maker opts to completely exclude it from the set of potential discount factors. 


Importantly, when $\delta = 1$, this case represents a discount structure that treats all time periods with equal importance.\footnote{The situation where $\delta = 1$ cannot be represented using a discounting sequence.} According to the variational discounting criterion, the decision maker attributes the maximum conceivable penalty to this discount structure, as indicated by $c(1) = \infty$, and as a result, she entirely disregards it.

The variational discounting representation is very broad, primarily because the conditions placed on the cost function $c$ are weak. For instance, consider the case in which the cost function $c$ assigns a cost of zero to all discount factors within a specific set $E \subseteq [0,1)$, while assigning an infinite cost to discount factors falling outside this set. Intuitively, the decision maker subjectively posits a set of discount factors that she deems accurate and imposes the maximal penalties on other discount factors. Furthermore, she exercises caution by adhering to the maxmin-type criterion: $$I(x)=\min_{\delta \in E} \bigg\{(1-\delta) \sum_{t=0}^{+\infty} \delta^t  x_t\bigg\}.$$

In some situations, the decision maker might exhibit full confidence, considering a single discount factor plausible while maximally penalizing all other discount factors. In such cases, the decision maker's preference is modeled by the EDU model
$$I(x)= (1-\delta) \sum_{t=0}^{+\infty} \delta^t  x_t$$ for some discount factor $\delta \in [0,1)$.

\section{Main results}\label{section3}

\subsection{Standard axioms}\label{section3.1}

The following axioms are well known in the literature:

\begin{axiom}[\textbf{Weak order}]\label{axiom1}
The relation $\succsim$ is a \textit{weak order}, i.e., it is complete (for every $x,y \in l_{\infty}$, either $x\succsim y$ or $y \succsim x$) and transitive (for every $x,y,z \in l_{\infty}$, if $x\succsim y$ and $y\succsim z$, then $x\succsim z$).
\end{axiom}

\begin{axiom}[\textbf{Monotonicity}]\label{axiom2}
(1) For every $x,y \in l_{\infty}$, if $x \geq y$, then $x\succsim y$, and (2) $1 \succ 0$. 
\end{axiom}

\begin{axiom}[\textbf{Continuity}]\label{axiom3}
For every $x,y,z \in l_{\infty}$ with $x\succsim y \succsim z$, the sets $\{\alpha \in [0,1]: \alpha x+ (1-\alpha)z \succsim y\}$ and $\{\alpha \in [0,1]: y \succsim \alpha x+ (1-\alpha)z\}$ are closed in $[0,1]$.
\end{axiom}

\begin{axiom}[\textbf{Invariance with respect to a common reference point (ICRP)}]\label{axiom4}
For every $x,y \in l_{\infty}$, $x \succsim y$ implies $x+ \theta \succsim y+ \theta$ for every $\theta \in \mathbb{R}$.
\end{axiom}

\begin{axiom}[\textbf{Convexity}]\label{axiom5}
For every $x,y \in l_{\infty}$, $\theta \in \mathbb{R}$, and $\lambda \in [0,1]$, if $x\succsim \theta$ and $y\succsim \theta$, then $\lambda x+(1-\lambda)y\succsim \theta$.
\end{axiom}

\begin{axiom}[\textbf{Monotone continuity}]\label{axiom8}
Let $x\in l_{\infty}$ and $\theta \in \mathbb{R}$ such that $x\succ \theta$, and let $(E_n)_{n\in \mathbb{N}}$ be a sequence of subsets of $\mathbb{N}$ with $(E_n) \downarrow \emptyset$. Then for every $k \in \mathbb{R}$, there exists $n_0 \in \mathbb{N}$ such that $k E_{n_0}x\succ \theta$.
\end{axiom}

\begin{axiom}[\textbf{Invariance with respect to the scale of utils (ISU)}]\label{axiom6}
For every $x,y\in l_{\infty}$ and $\alpha \geq 0$, $x\succsim y$ implies $\alpha x\succsim \alpha y$.
\end{axiom}

\begin{axiom}[\textbf{Invariance with respect to individual origins of utilities (IOU)}]\label{axiom7}
For every $z \in l_{\infty}$, if $x \sim y$, then $x+z \sim y+z$.
\end{axiom}

The first three axioms are standard. The first statement in Monotonicity means that if the decision maker obtains greater utility from a stream $x$ in any time period compared to another stream $y$, she should exhibit a weak preference for stream $x$ over $y$, and the second component of Monotonicity ensures that not all sequences are considered equivalent. Finally, Continuity is a technical assumption that guarantees a certain level of stability in preferences.

Convexity corresponds to a preference for equalizing or ``smoothing" utilities across different time periods. To illustrate this, consider a situation where we are indifferent between the two sequences $(1, -1, 1, -1, \ldots)$ and $(-1, 1, -1, 1, \ldots)$. Then the mixed utility stream given by $$\frac{1}{2}(1,-1,1,-1,\dots)+\frac{1}{2} (-1,1,-1,1,\dots)=(0,0,0,\dots)$$ should be considered preferable.

Countable additive discount structures streamline the presentation and find pervasive application across much of intertemporal choice. \textit{Monotone continuity} is a technical assumption introduced by \cite{V1964} and \cite{A1970} that allows us to establish representations involving standard discounting sequences. In decision theory, this axiom plays a crucial role in guaranteeing the countable additivity of beliefs in standard models (see \cite{A1970}, \cite{CMMJ2005}, and \cite{M2006}).

The axioms of ICRP, ISU, and IOU are introduced and interpreted in an intergenerational setting by \cite{CE2018}. As their name suggests, these axioms can be understood as invariance conditions. In essence, they express the idea that if two utility streams are connected in a certain manner, then for a specified class of transformations applied to those streams, the relationship between the two streams remains unchanged. ICRP and ISU together are equivalent to \textit{Co-cardinality} introduced by \cite{CE2018}. This assumption stipulates that the preference order between two streams remains unaltered when subjected to positive affine transformations of the streams. On the other hand, IOU necessitates a different class of transformations: adding a common stream to any given pair of streams should not change their ranking.

We say that $I$ is \textit{normalized} if $I(1)=1$, and $I$ is \textit{translation invariant} if for every $x \in l_{\infty}$ and $\theta \in \mathbb{R}$, $I(x+\theta) = I(x)+\theta$. The following proposition recalls under which conditions a constant equivalent of $\succsim$, which possesses some derived properties, exists (for a proof, see Proposition 1 in \cite{A2004}):

\begin{proposition}\label{pro1} 
A binary relation $\succsim$ on $l_{\infty}$ satisfies Axioms \ref{axiom1}-\ref{axiom5} if and only if there exists a unique weakly-increasing, concave, $1$-lipschitz, normalized, and translation invariant constant equivalent $I:l_{\infty}\rightarrow \mathbb{R}$ that represents $\succsim$. 
\end{proposition}

\subsection{Improving sequences and main invariance axiom} 


Our idea lies in the approach to formulating invariance properties: instead of expressing invariance properties in terms of transformations directly applied to a utility stream, we formulate them in terms of transformations applied to potential improvements from the stream.
\begin{definition} 
For every stream of utilities $x \in l_{\infty}$, an improving sequence from $x$ is a sequence $d \in l_{\infty}$ such that $x+d \succsim x$.  
\end{definition} 

Intuitively, we can interpret the sequence $x$ as representing the utility stream generated by the current status quo. Now, let us consider the adoption of a new long-term project (or policy) that introduces an additional utility sequence, denoted as $d$. In essence, when the project is implemented, the resulting utility stream becomes $x + d$. We can classify the sequence $d$ as an improving sequence if the project is deemed to be profitable. Let us state our main invariance principle:

\begin{linvashift}
For all $x,d\in l_{\infty}$, if $x+d \succsim x$, then $x+(0,d) \succsim x$.
\end{linvashift}

Our invariance principle posits that if $d$ is an improving sequence from $x$, then its delayed version $(0, d)$ remains an improving sequence from $x$. This principle captures the concept that when a new policy generates benefits, this profitability persists even when the policy's implementation is delayed by one period.

\subsection{Behavioral characterizations}


The variational discounting representation is characterized by standard Axioms \ref{axiom1}-\ref{axiom8}, along with our invariance principle IDIS.

\begin{theorem} \label{corollary1}
A binary relation $\succsim$ on $l_{\infty}$ satisfies Axioms \ref{axiom1}-\ref{axiom8} and IDIS if and only if it admits a variational discounting representation.

\end{theorem}


The complete proof can be found in Appendix \ref{proof-coro1}. Here, we will provide the proof idea for the ``only-if" part. Consider the set $F$, defined as $F= \{(1-\delta)(1,\delta,\delta^2,\dots): \delta \in [0,1)\}$, encompassing all possible geometric distributions. By the generalized Moreau duality (Appendix \ref{moreau}), in order to establish our representation theorem, it suffices to prove that, for any given $x \in l_{\infty}$, the \textit{$F$-superdifferential} of the functional $I$ at $x$, which coincides with $\partial I(x) \cap F$, is non-empty.


Under Axioms \ref{axiom1}-\ref{axiom5}, it can be shown that the constant equivalent function $I$ associated with the preference relation $\succsim$ possesses the following properties: it is weakly-increasing, concave, $1$-lipschitz, normalized, and translation invariant (as established in Proposition \ref{pro1}). Consequently, we can prove that the superdifferential of the function $I$ at any point $x$ is non-empty and is a subset of $\Delta$. Furthermore, Axiom \ref{axiom8} ensures that the superdifferential of $I$ at $x$ is indeed a subset of $\Delta^\sigma$. 

If $(1,0,0,\dots) \in \partial I(x)$, then it is evident that $\partial I(x) \cap F \neq \emptyset$. Now, we consider the case where $(1,0,0,\dots) \notin \partial I(x)$. In particular, the function $s: \partial I(x) \to \Delta^{\sigma}$ defined by $$s(p)= \frac{1}{1-p_0}(p_1,p_2,\dots)$$ is well defined. The crux of our proof lies in establishing that if $p\in \partial I(x)$, then $s(p) \in \partial I(x)$. This result is a consequence of our primary invariance principle. Hence, the mapping $s$ is continuous, operating from the convex and compact set $\partial I(x)$ to itself. Applying the Schauder–Tychonoff fixed-point theorem, we can deduce the existence of a fixed point $p$ in $\partial I(x)$, implying that $p \in F$. Consequently, for all utility streams $x$, the $F$-superdifferential at $x$ is non-empty, which concludes the proof.

\medskip

The introduction of the ISU assumption results in the cost function in Theorem \ref{corollary1} attaining values of either 0 or $\infty$. This condition gives rise to the \textit{maxmin discounting} criterion, as outlined in the following result:

\begin{proposition} \label{corollary2}
A binary relation $\succsim$ on $l_{\infty}$ satisfies Axioms \ref{axiom1}-\ref{axiom6} and IDIS if and only if there exists a unique closed set (in $[0,1]$) $E \subseteq [0,1)$ such that the constant equivalent representing $\succsim$ can be expressed as
$$I(x)=\min_{\delta \in E} \bigg\{(1-\delta) \sum_{t=0}^{+\infty} \delta^t  x_t\bigg\}.\footnote{In the special case when $\delta = 0$, we interpret this expression as the measure $(1,0,0,\dots)$, by means of the notational convention that $0^0=1$.}$$
\end{proposition}

\begin{remark}\label{rek1}
If Monotonicity is strengthened into \textit{Strong Monotonicity}, then the subjective set of discounted factors $E$ is bounded away from 0. Strong Monotonicity is defined as: $x \geq y$ implies $x\succsim y$, and $x > y$ implies $x \succ y$. Indeed, assuming by contradiction that $0 \in E$, we encounter the situation where $I((0,1,0,\dots)) = I((0,2,0,\dots)) = 0$, which clearly contradicts Strong Monotonicity.
\end{remark}

Lastly, if the decision maker is fully confident and, as a result, employs a single discount factor to evaluate utility streams, then she follows the standard EDU model.

\begin{proposition}\label{corollary-exponential}
A binary relation $\succsim$ on $l_{\infty}$ satisfies Axioms \ref{axiom1}-\ref{axiom7}, Strong monotonicity, and IDIS if and only if there exists a unique $\lambda \in (0,1)$ such that the constant equivalent representing $\succsim$ can be expressed as
$$I(x)= (1-\delta) \sum_{t=0}^{+\infty} \delta^t  x_t.$$
\end{proposition}



\subsection{Aggregation of exponential discounting experts}

We now assume that the decision maker makes intertemporal decisions with aid from a group of experts, each of whom follows the EDU model but with varying discount factors. Let $D \subseteq (0,1)$ be the set of discount factors recommended by these experts. Suppose $D$ is closed. Moreover, we assume that the decision maker's preference $\succsim$ satisfies all axioms characterizing the variational discounting criterion. In particular, there exists a grounded and lower semicontinuous function $c:[0,1] \to [0,\infty]$, where $c(1) = \infty$, such that the constant equivalent representing $\succsim$ can be represented by
$$I(x)=\min_{\delta \in [0,1)} \bigg\{(1-\delta) \sum_{t=0}^{+\infty} \delta^t  x_t + c(\delta)\bigg\}.$$

The following axiom appears in \cite{CE2018}, indicating that if all experts unanimously agree that a sequence $x$ is at least as preferable as a sequence $y$, then the decision maker should also adopt this preference ranking:
\begin{pareto}
For every $x,y \in l_{\infty}$, if $(1-\delta) \sum_{t} \delta^t x_t \geq (1-\delta) \sum_{t} \delta^t y_t$ for all $\delta \in D$, then $x\succsim y$.
\end{pareto}

\begin{proposition}\label{propaggre}
Suppose that $\succsim$ is represented by the variational discounting criterion with the cost function $c$. Then the following statements are equivalent:
\begin{enumerate}[(i)]
    \item Unanimity is satisfied.
    \item For all $\delta \notin D$, $c(\delta) = \infty$. Consequently, we can write $\succsim$ by $$I(x)=\min_{\delta \in D} \bigg\{(1-\delta) \sum_{t=0}^{+\infty} \delta^t  x_t + c(\delta)\bigg\}.$$
\end{enumerate}
\end{proposition}

Our aggregation rule implies that the decision maker should exclusively take into account the discount factors endorsed by experts. Each $\delta \in D$ is to be understood as a discount factor deemed plausible by some expert. The cost function $c$ can now be interpreted as a measure of the decision maker's confidence in these experts. For instance, when $c(\delta) = 0$, it signifies that the decision maker has complete confidence in the experts advocating for $\delta$. Conversely, if $c(\delta) = \infty$, it implies that the decision maker does not trust the experts and consequently ignores their advice.

The maxmin rule introduced by \cite{CE2018} is defined as follows: $$I(x)=\min_{\delta \in E} \bigg\{(1-\delta) \sum_{t=0}^{+\infty} \delta^t  x_t\bigg\},$$ where $E$ is a closed subset of $D$. This approach underscores two ideas: (i) that any expert could potentially be correct, and (ii) that it is justified to be conservative in assessing public projects. It is clear that the maxmin rule is a specific case of our ``variational" rule. Specifically, the maxmin rule can be derived if the decision maker's preference satisfies ISU.

\section{Generalization} \label{section4}




\subsection{Invariance with respect to linearly transforming improving sequences}

Notice that if we define the \textit{delayed} transformation $T_{de}:l_{\infty} \to l_{\infty}$ as $T_{de}(x)=(0,x)$ for every $x \in l_{\infty}$, then IDIS can be stated as follows: for all $x,d\in l_{\infty}$, $x+d \succsim x$ implies that $x+T_{de} (d) \succsim x$. Clearly, $T_{de}$ is a continuous positive linear function. In this subsection, we investigate the representation of $\succsim$ when the delayed transformation in IDIS is substituted with any continuous positive linear transformation.

Let us consider a continuous positive linear function $T: l_{\infty} \to l_{\infty}$. The next axiom, called \textit{Invariance with respect to applying $T$ to improving sequences}, asserts that $T$ maps improving sequences from $x$ to other improving sequences from $x$. The principle of I$T$IS is similar to IDIS, yet it accommodates a broader range of operators $T$.

\begin{linva}
For every $x,d\in l_{\infty}$, if $x+d \succsim x$, then $x+T(d) \succsim x$.
\end{linva}

Before presenting our findings, we must first introduce some mathematical definitions and notations.

\begin{definition} 
The \textit{adjoint} transformation of $T$ is a continuous positive linear function $T^*: ba(\mathbb{N}) \to ba(\mathbb{N})$ that satisfies the following conditions for every $x \in l_{\infty}$ and $\mu \in ba(\mathbb{N})$: 
\begin{equation*} 
\langle x,T^*(\mu) \rangle = \langle T(x),\mu \rangle.
\end{equation*} 
\end{definition} 

The next result establishes that, for any continuous positive linear function $T: l_{\infty} \to l_{\infty}$, its adjoint transformation not only exists but is also uniquely defined. 

\begin{lemma}\label{lemdefadj}
For any continuous positive linear function $T: l_{\infty} \to l_\infty$, there exists a unique adjoint transformation of $T$.
\end{lemma}

Assume that we are using a discount structure represented by $\mu \in ba(\mathbb{N})$ to evaluate the value of a utility stream $x$. This evaluation yields a total utility $\langle T(x), \mu \rangle$. There are two methods for distorting the total utility: one method involves distorting the utility stream $x$, while the other distorts the discount structure $\mu$. In terms of the total utility generated, these two distortions are equivalent. 


For any continuous positive linear function $T^*: ba(\mathbb{N}) \to ba(\mathbb{N})$, we define an \textit{eigenvector} of $T^*$ as a vector $p\in ba(\mathbb{N})$ that satisfies $T^*(p) = \lambda p$ for some $\lambda \in \mathbb{R}$. We define $\Delta (T^*)$ as the set of discount structures that are eigenvectors of $T^*$, i.e., $\Delta (T^*)=\{p \in \Delta: \text{$p$ is an eigenvector of $T^*$}\}$. We remark that this set is non-empty and closed, as established in Lemma \ref{remark1}.

\begin{definition}
For a given continuous positive linear function $T^*: ba(\mathbb{N}) \to ba(\mathbb{N})$, a preference relation $\succsim$ on $l_\infty$ is a \textit{$T^*$-variational time preference} if there exists a grounded and lower semicontinuous function $c: \Delta (T^*) \rightarrow [0,\infty]$ such that, for every $x\in l_{\infty}$, the constant equivalent of $x$ is given by
\begin{equation*} 
I(x)=\min_{p \in \Delta (T^*)}\{\langle x,p \rangle + c(p)\}.
\end{equation*}
\end{definition}

In the $T^*$-variational representation, the set of possible discount structures $\Delta (T^*)$ does not include all discount structures. Rather, it imposes a requirement that each discount structures $p \in \Delta (T^*)$ must exhibit an ``invariance" to the linear transformation $T^*: ba(\mathbb{N})\to ba(\mathbb{N})$. This invariance implies that when $T^*$ is applied to $p$, it does not change the direction of the discount structure; instead, it only scales it by a positive scalar factor $\lambda$. In other words, applying $T^*$ to $p$ does not impact the approach to discounting time periods because evaluating utility sequences using either $p$ or $T^* (p)$ results in the same optimal choices.

\begin{theorem} \label{theorem1}
If a binary relation $\succsim$ satisfies Axioms \ref{axiom1}-\ref{axiom5} and I$T$IS, then it admits a $T^*$-variational time representation. Conversely, if $\succsim$ is a $T^*$-variational time preference, then it satisfies Axioms \ref{axiom1}-\ref{axiom5} and I$T$IS provided that $\langle 1, T^*(p) \rangle \leq 1$ for every $p \in \Delta (T^*)$. 
\end{theorem}

The proof of this theorem shares similarities with that of Theorem \ref{corollary1}. We will only provide an outline of the proof for the first part. We first remark that the set $\Delta (T^*)$ is non-empty (Lemma \ref{remark1}). Using the generalized Moreau duality (see Appendix \ref{moreau}), our goal is to prove that, for every $x \in l_{\infty}$, the $\Delta (T^*)$-superdifferential at $x$ is non-empty. Clearly, the superdifferential of the functional $I$ at point $x$ is non-empty subset of $\Delta$. By definition, $p \in \Delta (T^*)$ if and only if either $T^* (p) = 0$ or $p$ is a fixed point of the \textit{normalized} transformation defined by $s(p')=\frac{T^*(p')}{\Vert T^* (p') \Vert}$ (whenever this transformation exists). We will now consider two cases.

\textbf{Case 1:} If there exists $p \in \partial I(x)$ such that $T^*(p) = 0$, it follows that $p \in \Delta$ since $\partial I(x) \subseteq \Delta$. By definition, we get $p \in \Delta (T^*)$. Consequently, $p$ is in the $\Delta (T^*)$-superdifferential at $x$.

\textbf{Case 2:} If $T^*(p) \neq 0$ for all $p\in \partial I(x)$, then the function $s$ is well defined on $\partial I(x)$. We can show that, for any $p\in \partial I(x)$, the transformation $s$ maps $p$ back into $\partial I(x)$. Thus, $s$ is a continuous mapping from the convex compact set $\partial I(x)$ to itself. By the Schauder–Tychonoff fixed-point theorem, we conclude that the function $s$ has a fixed point, denoted as $p$, in $\partial I(x)$. This finding further implies that $p \in \Delta (T^*)$. Thus, for all $x \in l_{\infty}$, the $\Delta (T^*)$-superdifferential at $x$ is non-empty, thereby completing the proof of the first part.

\begin{proposition}\label{promaxmin}
If $\succsim$ is a binary relation on $l_{\infty}$ satisfying Axioms \ref{axiom1}-\ref{axiom5}, Axiom \ref{axiom6}, and I$T$IS, then there exists a non-empty closed set $D \subseteq \Delta (T^*)$ such that the unique constant equivalent $I:l_{\infty}\rightarrow \mathbb{R}$ representing $\succsim$ can be expressed as
\begin{equation*} 
I(x)=\min_{p \in D}\langle x,p \rangle.
\end{equation*} 

Conversely, if $I(x)=\min_{p \in D} \langle x,p \rangle$ for some non-empty closed set $D \subseteq \Delta (T^*)$, then $\succsim$ represented by $I$ satisfies Axioms \ref{axiom1}-\ref{axiom6} and I$T$IS provided that $\langle 1, T^*(p) \rangle \leq 1$ for every $p \in \Delta (T^*)$. 
\end{proposition}

The significance of the assumption ``$\langle 1, T^*(p) \rangle \leq 1$ for every $p \in \Delta(T^*)$" to establish the inverse implication in the above results is illustrated in the next example. 


\begin{example}\label{example2}
Consider the transformation $T (d) = 2d$ for every $d \in l_{\infty}$. It is evident that $T^*(\mu) = 2 \mu$ for every $\mu \in ba(\mathbb{N})$. Thus, the sets of discount structures defined in Theorem \ref{theorem1} and Proposition \ref{promaxmin} are $\Delta$.

Now, let $\succsim$ be the binary relation on $l_{\infty}$ represented by $I(x) = \min_{p\in \Delta} \langle x,p\rangle = \inf_{t\in \mathbb{N}} x_t$. Take $x= (-1,0,0,\dots)$ and $d = (1,-1,0,0,\dots)$. We observe that $I(x+d) = I(x) = -1$, yet $-1= I(x) > I(x+2d) = -2$. This implies that $\succsim$ does not satisfy I$T$IS.
\end{example}

\subsection{A family of transformations}

In this subsection, we consider a preference relation that exhibits invariance with respect to a family of transformations. We show that the set of discount structures in the maxmin criterion can be narrowed down even further. Specifically, it can be constrained to the intersection of the sets of eigenvectors associated with the adjoints of these transformations. Unfortunately, we do not have an analogous result for the case of variational representation. Let $(T_{i})_{i\in K}$ be a collection of transformations.

\begin{theorem}\label{theorem2.3}
If $\succsim$ satisfies Axioms \ref{axiom1}-\ref{axiom5}, Axiom \ref{axiom6}, and I$T_i$IS for every $i\in K$, then there exists a non-empty closed set $D \subseteq F$ such that the unique constant equivalent $I:l_{\infty}\rightarrow \mathbb{R}$ representing $\succsim$ can be expressed as
\begin{equation*} 
I(x)=\min_{p \in D}\langle x,p \rangle,
\end{equation*} 
where $F=\bigcap_{i\in K} \{p \in \Delta: \text{$p$ is an eigenvector of $T_i^{*}$} \}$.

Conversely, if $I(x)=\min_{p \in D} \langle x,p \rangle$ for some non-empty closed set $D \subseteq F$, then $\succsim$ represented by $I$ satisfies Axioms \ref{axiom1}-\ref{axiom5}, Axiom \ref{axiom6}, and I$T_i$IS for every $i\in K$ provided that $\langle 1, T_i^*(p) \rangle \leq 1$ for every $p \in F$ and $i\in K$. 
\end{theorem}

\subsection{Extension to countable additivity}

By adding Monotone continuity, we can obtain representations with standard discounting sequences. Let $A$ be a subset of a topological space, we say that a function $f:\overline{A} \to [0,\infty]$ is \textit{$A^c$-infinite} if $f(x) = \infty$ for every $x \in \overline{A} \setminus A$, where $\overline{A}$ is the closure of $A$. That is, $f$ takes the value of infinity outside of $A$.

\begin{theorem} \label{theorem2}
If $\succsim$ is a binary relation on $l_{\infty}$ satisfying Axioms \ref{axiom1}-\ref{axiom8} and I$T$IS, then there exists a grounded, $F^c$-infinite, and lower semicontinuous function $c: \overline{F} \rightarrow [0,\infty]$ such that the unique constant equivalent $I:l_{\infty}\rightarrow \mathbb{R}$ representing $\succsim$ can be expressed as
\begin{equation} \label{constanteq}
I(x)=\min_{p \in F}\bigg\{\sum_{n=0}^{\infty} p_n x_n  + c(p)\bigg\},
\end{equation} 
where $F=\{p \in \Delta^{\sigma}: \text{$p$ is an eigenvector of $T^*$}\}$.

Conversely, if $I(x)=\min_{p \in F}\{\sum_{n=0}^{\infty} p_n x_n  + c(p)\}$ for some grounded, $F^c$-infinite, and lower semicontinuous function $c: \overline{F} \rightarrow [0,\infty]$, then $\succsim$ represented by $I$ satisfies Axioms \ref{axiom1}-\ref{axiom8} and I$T$IS provided that $\langle 1, T^*(p) \rangle \leq 1$ for every $p \in F$. 
\end{theorem}

\begin{remark}
The set $F$ defined in the theorem above is not always closed in $ba(\mathbb{N})$. For example, when $T$ is given by Example \ref{example2}, then $F = \Delta^{\sigma}$, which is not closed. Thus, the constant equivalent defined by (\ref{constanteq}) might not be well defined. However, because $c$ is $F^c$-infinite, we can write
\begin{equation*} 
I(x)=\min_{p \in \overline{F}}\bigg\{\sum_{n=0}^{\infty} p_n x_n  + c(p)\bigg\}.
\end{equation*} 
Since $c$ is lower semicontinuous, $I$ is well defined. 
\end{remark}

\begin{proposition} \label{prop-maxmin2}
If $\succsim$ is a binary relation on $l_{\infty}$ satisfying Axioms \ref{axiom1}-\ref{axiom6} and I$T$IS, then there exists a non-empty closed set (in $\Delta$) $D \subseteq F$ such that the unique constant equivalent $I:l_{\infty}\rightarrow \mathbb{R}$ representing $\succsim$ can be expressed as
\begin{equation*} 
I(x)=\min_{p \in D} \sum_{n=0}^{\infty} p_n x_n,
\end{equation*} 
where $F=\{p \in \Delta^{\sigma}: \text{$p$ is an eigenvector of $T^*$} \}$.

Conversely, if $I(x)=\min_{p \in D} \langle x,p \rangle$ for some non-empty closed set (in $\Delta$) $D \subseteq F$, then $\succsim$ represented by $I$ satisfies Axioms \ref{axiom1}-\ref{axiom6} and I$T$IS provided that $\langle 1, T^*(p) \rangle \leq 1$ for every $p \in F$. 
\end{proposition}


\section{Patience}\label{sec:equal}


As discussed in the introduction, there are numerous arguments supporting the theory of not discounting future utilities. In this section, we will introduce criteria that covers situations where the decision maker may consider discount structures that treat all time periods equally. We first remark that if a countably additive discount structure $p = (p_0, p_1, \dots)$ is meant to represent the equal treatment of all time periods, then we would need $p_i = p_j$ for all $i, j \in \mathbb{N}$. However, it is clear that there is no countably additive discount structure that satisfies this requirement. Thus, the principle of equal treatment for all time periods cannot be adequately represented by countably additive discount structures. Consequently, it becomes essential to remove Monotone Continuity, as we shall proceed to do.

In the literature, the principle of ensuring equal treatment of all time periods is often described by \textit{Patience}. This axiom dictates that preferences should remain unchanged when the utility levels of a finite number of time periods are permuted.\footnote{A finite permutation $\sigma: \mathbb{N} \to \mathbb{N}$ is a bijection where $\sigma(i)=i$ for all integers $i$ except for a finite number of them.} For any finite permutation $\sigma:\mathbb{N} \to \mathbb{N}$, we denote $x_{\sigma}=(x_{\sigma(0)},x_{\sigma(1)},x_{\sigma(2)},\dots)$ for all $x \in l_{\infty}$.

\begin{patience}
For every $x \in l_{\infty}$ and every finite permutation $\sigma: \mathbb{N} \to \mathbb{N}$, $x \sim x_{\sigma}$.
\end{patience}

A discount structure $\mu \in \Delta$ is called a \textit{Banach-Mazur limit} if, for every $x \in l_{\infty}$, $\langle x,\mu \rangle = \langle (x_1,x_2,\dots),\mu \rangle$. Banach-Mazur limits exhibit a shift-invariance property, meaning they assess a utility stream $x$ and its shifted version $(x_1, x_2, x_3, \dots)$ equally. Clearly, if $\mu$ is a Banach-Mazur limit, then $\langle x,\mu \rangle= \langle x_{\sigma}, \mu \rangle$ for all finite permutation $\sigma$. Consequently, Banach-Mazur limits embody the principle of treating all time periods equally. We use $\mathcal{B}$ to denote the set of all Banach-Mazur limits. 

\begin{theorem}\label{theoremequal}
A binary relation $\succsim$ on $l_{\infty}$ satisfies Axioms \ref{axiom1}-\ref{axiom5} and IDIS if and only if there exists a grounded and lower semicontinuous function $c: \mathcal{D} \to [0,\infty]$ such that the constant equivalent of $\succsim$ can be expressed as
\begin{equation*} 
I(x)=\min_{p \in \mathcal{D}} \{\langle x,p \rangle +c(p)\},
\end{equation*} where $\mathcal{D} = \mathcal{B} \cup \{(1-\delta)(1,\delta,\delta^2,\dots):\delta \in [0,1)\}$.
\end{theorem}


This outcome suggests that, in addition to considering exponential discounting sequences, the decision maker also take into account Banach-Mazur limits, which possibly captures patient behavior.

Typically, Patience is substituted with a stronger assumption known as \textit{Time invariance}. This axiom means that the decision maker assigns zero weight to consequences in all prior and current time periods, while assigning full weight to future time periods. In essence, the outcomes at any finite set of points in time become irrelevant. A crucial aspect of Time invariance is that it is a requisite condition for any preference relation aiming to represent the undiscounted case as a limit case of discounting as the discount factor approaches 1 (\cite{M1998}). 

\begin{timeinva}
For every $x\in l_{\infty}$, $x\sim (x_1,x_2,x_3,\dots)$.
\end{timeinva}

If we introduce the assumption of Time invariance in addition to the existing conditions, then the set of potential discount structures is constrained to include only Banach-Mazur limits.

\begin{proposition}\label{corollaryBanach}
A binary relation $\succsim$ on $l_{\infty}$ satisfies Axioms \ref{axiom1}-\ref{axiom5}, IDIS, and Time variance if and only if there exists a unique grounded, lower semicontinuous, and convex function $c: \mathcal{B} \to [0,\infty]$ such that the constant equivalent of $\succsim$ can be expressed as
\begin{equation*} 
I(x)=\min_{p \in \mathcal{B}} \{\langle x,p \rangle +c(p)\}.
\end{equation*}
\end{proposition}

In the work of \cite{M1998}, the following criterion is introduced and characterized to encapsulate Patience: $$I(x) = \liminf_{T \uparrow \infty} \inf_{j \geq 0} \frac{1}{T+1} \sum_{t=0}^{T} x_{j+t}.$$ Notably, \cite{KS2018} demonstrate that this criterion can be equivalently expressed using the following maxmin criterion:
$$I(x) = \min_{p\in \mathcal{B}} \langle x,p \rangle.$$ Therefore, Marinacci's patient criterion can be regarded as a special case of the criterion outlined in Proposition \ref{corollaryBanach}.

\begin{remark}
We can establish the uniqueness of the cost function $c$ in the preceding result because the set $\mathcal{B}$ is both closed and convex. This proof will mirror the argument presented in Appendix \ref{appenunique}. 
\end{remark}



Now, let us introduce a new patience condition.

\begin{fpermuinva}
For every finite permutation $\sigma:\mathbb{N}\to \mathbb{N}$ and $x,d\in l_{\infty}$, if $x+d\succsim x$, then $x+d_{\sigma} \succsim x$.
\end{fpermuinva}

It is important to note that, unlike the invariance axiom I$T$IS which necessitates invariance with respect to a single transformation, IFPIS introduces a requirement for invariance across an infinite family of transformations. This assumption could be interpreted as a patient condition: when dealing with an improving sequence, it is acceptable to exchange the increments of any two time periods, while leaving those of the remaining dates unchanged. The key distinction between Patience and IFPIS is that Patience requires the consideration of permutations applied to actual situations, while IFPIS demands the same treatment for permutations applied to improving sequences.

A finitely additive measure $\mu$ is \textit{purely finitely additive} if, for every $n \in \mathbb{N}$, $\mu(\{n\})=0$. We denote $pa(\mathbb{N})$ the set of all purely finitely additive measures.

\begin{theorem}\label{corollary7}
A binary relation $\succsim$ on $l_{\infty}$ satisfies Axioms \ref{axiom1}-\ref{axiom5}, Axiom \ref{axiom6}, and IFPIS if and only if there exists a unique convex and compact set $D\in pa(\mathbb{N}) \cap \Delta$ such that the unique constant equivalent of $\succsim$ can be expressed as
\begin{equation*} 
I(x)=\min_{\mu \in D } \langle x,\mu \rangle.
\end{equation*}  
\end{theorem} 

\begin{remark}\label{remark4}
\cite{R2007} shows that under Axioms \ref{axiom1}-\ref{axiom5}, Axiom \ref{axiom6}, and Patience, the unique constant equivalent of $\succsim$ is represented by the criterion stated in Theorem \ref{corollary7}. Thus, Patience is equivalent to IFPIS.

By incorporating the Axiom of IOU into Theorem \ref{corollary7}, it leads to the existence of a unique purely finitely additive discount structure denoted as $\mu$, such that the constant equivalent of $\succsim$ can be expressed as $I(x)= \int x d\mu$. Likewise, if we further incorporate ISU and IOU into Proposition \ref{corollaryBanach}, the discount structure $\mu$ must be a Banach-Mazur limit. Note that both additive criteria are introduced by \cite{R2007}.
\end{remark}

Now, let us explore the implications of strengthening IFPIS into what we will call \textit{Invariance with respect to permutation of improving sequences}. The stronger axiom demands that improving sequences must remain unaltered under any permutation. However, this axiom cannot coexist with Axioms \ref{axiom1}-\ref{axiom5} and Axiom \ref{axiom6}.

\begin{permuinva}
For every permutation $\sigma:\mathbb{N}\to \mathbb{N}$ and $x,d\in l_{\infty}$, if $x+d\succsim x$, then $x+T_{\sigma}(d) \succsim x$.
\end{permuinva}

\begin{proposition}\label{corollary9}
There does not exist any constant equivalent that represents a binary relation on $l_{\infty}$ satisfying Axioms \ref{axiom1}-\ref{axiom5}, Axiom \ref{axiom6}, and IPIS.
\end{proposition}

\begin{remark}

A more stringent form of the Patience axiom, which is \textit{Naive Patience}, asserts that a stream $x$ is equivalent to any of its permutations. This axiom characterizes an agent endowed with ``infinite" patience, where all points in time, regardless of their remoteness, hold equal significance. We have shown that, under Axioms \ref{axiom1}-\ref{axiom5} and Axiom \ref{axiom6}, Patience and IFPIS are equivalent (Remark \ref{remark4}). The same conclusion does not hold for IPIS and Naive patience. Indeed, consider a preference relation $\succsim$ on $l_{\infty}$ represented by the liminf criterion $I(x) = \lim \inf x_t = \min_{p \in \Delta} \langle x, p\rangle$ (see \cite{M1998} and \cite{R2007}). Clearly, $\succsim$ satisfies Axioms \ref{axiom1}-\ref{axiom5}, Axiom \ref{axiom6}, and Naive patience. Hence, by Proposition \ref{corollary9}, it becomes evident that Naive patience is a less demanding condition compared to IPIS.

In fact, it can be demonstrated that $\succsim$ does not meet the conditions of IPIS. Let us consider the following two sequences: $$x=(0,1,0,1,0,\dots) \text{ and } d=(1,-1,1,-1,1,\dots).$$ 
The sequence $x+d$ results in $(1,0,1,0,1,\dots)$, where $I(x+d) = I(d) = 0$. However, by introducing a permutation defined as $\sigma(2i) = 2i+1$ and $\sigma(2i+1) = 2i$ for all $i \in \mathbb{N}$, we transform $d$ into $T_{\sigma}(d) = (-1,1,-1,1,-1,\dots)$. Consequently, $x+T_{\sigma}(d)$ becomes $(-1,2,-1,2,-1,\dots)$. This results in $-1 = I(x+T_{\sigma}(d)) < I(x)$, clearly showing that the preference relation $\succsim$ does not satisfy IPIS.
\end{remark}

\section{Conclusion}

In this paper, we present a novel criterion, known as variational discounting, which offers a robust solution for addressing the wide range of discount rates that decision makers may consider. Additionally, we present a model designed to adapt to scenarios in which decision makers deem it plausible not to discount future utilities. Our approach not only encompasses and extends several established criteria found in the literature but also introduces new perspectives in this domain. Moreover, our representations are based on a class of invariance principles that are different from those in the literature on intertemporal choice. Rather than expressing invariance properties through transformations directly applied to a utility stream, we formulate them in terms of transformations applied to potential improvements from the stream.


\section*{Appendix} 
\appendix

\section{Preliminaries}

\subsection{Yosida-Hewitt decomposition theorem}
\begin{yhtheorem}
For every $\mu \in ba(\mathbb{N})$, there exists a unique pair $(\mu_c,\mu_p) \in ca(\mathbb{N}) \times pa(\mathbb{N})$ such that $\mu = \mu_c+\mu_p$.
\end{yhtheorem}

Yosida-Hewitt decomposition theorem says that each finitely additive measure can be decomposed into a sum of a countably additive measure and a purely finitely additive measure, and this decomposition is unique.

\subsection{Schauder–Tychonoff fixed-point theorem}  \label{schauder}

The Schauder–Tychonoff fixed-point theorem is a fundamental result in functional analysis and topology that deals with the existence of fixed points for a broad class of mappings. This theorem is a generalization of the Brouwer fixed-point theorem and is named after the mathematicians Juliusz Schauder and Andrey Tychonoff. The theorem can be stated as follows:
\begin{Schau}
 Let $X$ be a non-empty, compact, convex subset of a topological vector space (often a Banach space), and let $f: X \to X$ be a continuous mapping. Then, there exists at least one fixed point for the mapping $f$, i.e., there exists some $x$ in $X$ such that $f(x) = x$.   
\end{Schau}

In this paper, we apply this theorem when $X$ is a non-empty, compact, convex subset of $ba(\mathbb{N})$ endowed with the weak$^{\star}$ topology on $ba(\mathbb{N})$.

\subsection{Eigenvectors and eigenvalues}\label{eigen}

If a non-zero set function $\mu \in ba(\mathbb{N})$ and a real scalar $\lambda$ satisfy the equation $T^*(\mu) = \lambda \mu$, we designate $\mu$ as an \textit{eigenvector} and $\lambda$ as the corresponding \textit{eigenvalue} of the operator $T^*$. Specifically, we refer to $\mu$ as an eigenvector of $T^*$ associated with the eigenvalue $\lambda$, and conversely, $\lambda$ as the (unique) eigenvalue of $T^*$ associated with $\mu$.

We call $p$ a \textit{normalized eigenvector} of $T^*$ if $p \in \Delta$ is an eigenvector of $T^*$. We denote $\Delta (T^*)=\{p \in \Delta: \text{$p$ is an eigenvector of $T^*$}\}$ the set of all normalized eigenvectors of $T^*$. The following result shows that $\Delta (T^*)$ is non-empty and closed.



\begin{lemma}\label{remark1}
The set $\Delta (T^*)$ is non-empty and closed.
\end{lemma}

\begin{proof}
Let $T^*: ba(\mathbb{N}) \to ba(\mathbb{N})$ be a continuous positive linear transformation. Now, if there exists $\mu \in \Delta$ such that $T^*(\mu)=0$, then $\mu$ is a normalized eigenvalue of $T^*$. Otherwise, assume that, for every $\mu \in \Delta$, $T^*(\mu) \neq 0$. Let $\mu \in \Delta$, then $T^* (\mu) \geq 0$ since $T^*$ is positive, and $T^* (\mu) \neq 0$. Observe that $\langle 1,T^*(\mu)\rangle$ is the variation norm of $ T^*(\mu)$, so $\langle 1,T^*(\mu)\rangle > 0$. Hence, we can define the operator $f: \Delta \to ba(\mathbb{N})$ as follows:
$$f(\mu) = \frac{T^*(\mu)}{\langle 1,T^*(\mu)\rangle}.$$ Moreover, for every $\mu \in \Delta$, $f(\mu) \geq 0$ and $\langle 1, f(\mu) \rangle = 1$, which implies that $f(\mu)\in \Delta$. Thus, we have $f: \Delta \to \Delta$.

By definition, the mapping $T^*$ is continuous on $\Delta$; and it is clear that the function $\mu \in ba_{+}(\mathbb{N}) \to \langle 1, \mu \rangle$ is also continuous from $ba_{+}(\mathbb{N})$ (endowed with the relative topology induced by the weak$^{\star}$ topology on $ba(\mathbb{N})$) to $\mathbb{R}$, where $ba_{+}(\mathbb{N})= \{\mu \in ba(\mathbb{N}): \mu \geq 0\}$. Thus, $f$ is a continuous function from a compact convex set to itself. Hence, from Schauder–Tychonoff fixed-point theorem (see Appendix \ref{schauder}), there exists $\mu \in \Delta$ such that $f(\mu) = \mu$, which is equivalent to $T^*(\mu)=\langle 1,T^*(\mu)\rangle \mu$, i.e., $\mu$ is a normalized eigenvector of $T^*$.  Thus, the set $\Delta (T^*)$ is non-empty.


Now, let $(p_d)_{d\in D}$ be a net in $\Delta (T^*)$ such that $p_d$ converges to $p \in \Delta$ (because $\Delta$ is compact). By definition, for every $d\in D$, there exists $\lambda_d \geq 0$ such that $T^*(p_d) = \lambda_d p_d$. Since $T^*$ is continuous, we have $T^*(p_d)$ converging to $T^*(p)$. Now, considering $\langle 1, T^*(p_d) \rangle = \langle 1, \lambda_d p_d \rangle = \lambda_d$, and noting that the set $\{\langle 1, T^*(p) \rangle: p \in \Delta\}$ is bounded, we conclude that $(\lambda_d)_{d\in D}$ is bounded. Consequently, it admits a subnet that converges to a point $\lambda \in [0, \infty)$. Passing to the limit, we obtain $T^*(p) = \lambda p$, which implies that $p \in \Delta (T^*)$. Therefore, the set $\Delta (T^*)$ is closed.
\end{proof}

The following definition and results introduce a practical way to simplify the computation of eigenvectors of adjoint transformations. 

\begin{definition} 
A linear transformation $T^*$ is called a Yosida-Hewitt transformation (YH transformation) if, for every $\mu \in ca(\mathbb{N})$, $T^*(\mu) \in ca(\mathbb{N})$, and, for every $\mu \in pa(\mathbb{N})$, $T^*(\mu) \in pa(\mathbb{N})$.
\end{definition} 
A YH transformation maps every countably additive measure (resp. purely finitely additive measure) to a countably additive measure (resp. purely finitely additive measure). We use the notation $\mathbb{I}_A$ to represent the sequence defined as $\mathbb{I}_A=xAy$, where $x=1$ and $y=0$.

\begin{lemma} \label{criteriu} 
If (i) $T (\mathbb{I}_{A_n})$ converges to 0 whenever $(A_n)\downarrow \emptyset$ (where $A_n \subseteq \mathbb{N}$ for all $n\in \mathbb{N}$), and (ii) for every $n\in \mathbb{N}$, there exists a finite set $E_n \subseteq \mathbb{N}$ such that $T (\mathbb{I}_{\{n\}})$ has its support in $E_n$, then its adjoint $T^*$ is a YH transformation.  
\end{lemma} 

\begin{proof}
Assume that the assumptions in the lemma are satisfied, and consider $\mu \in ca(\mathbb{N})$. If $(A_n)\downarrow \emptyset$, then $\langle \mathbb{I}_{A_n},T^*(\mu) \rangle=\langle T(\mathbb{I}_{A_n}),\mu \rangle$ converges to zero when $n$ tends to infinity, which proves that $T^*(\mu) \in ca(\mathbb{N})$. Now, let $\mu \in pa(\mathbb{N})$, then for every $n \in \mathbb{N}$, $\langle \mathbb{I}_{\{n\}},T^*(\mu) \rangle=\langle T(\mathbb{I}_{\{n\}}),\mu \rangle=0$ from assumption (ii), which finally ends the proof.
\end{proof}

We denote $E_{\lambda} (T^*)$ the set of eigenvectors of $T^*$ associated to $\lambda$. Moreover, we use $E^{ca}_{\lambda}(T^*)$ and $E^{pa}_{\lambda}(T^*)$ to denote the sets of eigenvectors of $T^*$ associated to $\lambda$ which are countably additive and purely finitely additive, respectively. Clearly, $E^{ca}_{\lambda}(T^*)$ and $E^{pa}_{\lambda}(T^*)$ are two topological vector spaces. We use $A \oplus B$ to denote the direct sum of two vector spaces $A$ and $B$.

\begin{lemma} \label{lemma3}
If $T^*$ is a YH transformation, then for every eigenvalue $\lambda$ of $T^*$, $E_{\lambda}(T^*) = E^{ca}_{\lambda}(T^*) \oplus E^{pa}_{\lambda}(T^*)$.
\end{lemma} 

\begin{proof}
Assume that $T^*$ is a YH transformation. Let $\lambda$ be an eigenvalue of $T^*$, and let $\mu \in ba(\mathbb{N})$ be an eigenvector of $T^*$ associated to $\lambda$. From Yosida-Hewitt decomposition theorem, $\mu=\mu_c+\mu_p$ where $\mu_c \in ca(\mathbb{N})$ and $\mu_p \in pa(\mathbb{N})$. We also have $T^*(\mu_c)+T^*(\mu_p) = T^*(\mu) = \lambda \mu_c+\lambda \mu_p$. Thus, from the uniqueness of the decomposition in Yosida-Hewitt theorem and the fact that $T^*$ is a YH transformation, we get that $T^*(\mu_c)=\lambda \mu_c$ and $T^*(\mu_p)=\lambda \mu_p$, i.e., $\mu_c \in E^{ca}_{\lambda}(T^*)$ and $\mu_p \in E^{pa}_{\lambda}(T^*)$. Finally, the lemma is proved because $E^{ca}_{\lambda}(T^*) \cap E^{pa}_{\lambda}(T^*)=\emptyset$.
\end{proof} 

\subsection{Generalized Moreau duality}\label{moreau}
In the context of a topological vector space $X$, with a non-empty subset $F$ of $\mathbb{R}^{X}$, and a function $\varphi: X\rightarrow \mathbb{R}$, we define the \textit{$F$-superdifferential} of $\varphi$ at a point $x\in X$ as follows:
\begin{align*}
\partial^{F}\varphi(x)=\{f \in F : \varphi (y)-\varphi(x)\leq \langle y,f\rangle-\langle x,f\rangle \text{ for every $y \in X$}\},
\end{align*}
where $\langle x,f \rangle = f(x)$. The \textit{$F$-conjugate} of $\varphi$, $\varphi^{F}:F\rightarrow [-\infty,\infty]$, is defined by 
\begin{align*}
\varphi^{F} (f)=\inf_{x \in X}\{\langle x,f \rangle -\varphi(x)\}.
\end{align*} 
The forthcoming lemma will play a central role in the proof of our main results:

\begin{lemma}\label{lem2}
If $\partial^{F}\varphi(x) \neq \emptyset$ for every $x \in X$, then 
\begin{align*}
&\varphi(x)=\min_{f\in F}\{\langle x,f \rangle - \varphi^{F}(f)\}, \text{ and} \\
&\partial^{F}\varphi(x) = \argmin_{f\in F}\{\langle x,f \rangle - \varphi^{F}(f)\}
\end{align*}
for every $x\in X$. Furthermore, $\varphi^{F}$ is the maximal function among all functions $R:F\rightarrow \mathbb{R}\cup \{-\infty\}$ satisfying
\begin{align*}
\varphi(x)=\min_{f\in F}\{\langle x,f \rangle - R(f)\}.
\end{align*}
\end{lemma}
\begin{proof}
See Lemma 4 in \cite{C2014}.
\end{proof}

When $F = X'$ (the topological dual of $X$), the operator $\partial^{F} \varphi$ is referred to as the \textit{superdifferential} of $\varphi$ and is denoted by $\partial \varphi$. Additionally, $\varphi^{F}$ is termed the \textit{Fenchel conjugate} of $\varphi$ and is denoted as $\varphi^{\star}$. Similarly, the conjugate of a function $\phi: X' \rightarrow \mathbb{R}$ is the function $\phi^{\star}: X \rightarrow [-\infty, \infty]$, defined as: $$\phi ^{\star} (x) = \inf_{x' \in X'} \{\langle x,x' \rangle - \phi (x')\}.$$ The following result is a standard theorem in duality theory:
\begin{bicon}
Let $X$ be a Hausdorff locally convex space. For any function $\varphi : X \rightarrow \mathbb{R}$, $\varphi = \varphi^{\star \star}$ if and only if $\varphi$ is a upper semicontinuous and concave function.
\end{bicon}

\section{Discussion on uniqueness of the cost function}\label{appenunique}

Let $F= \{(1-\delta)(1,\delta,\delta^2,\dots): \delta \in [0,1)\}$, and let $D$ be the closed convex hull of $F$. Suppose that the preference relation $\succsim$ can be represented using the following two variational discounting representations:
$$I(x)=\min_{\delta \in [0,1)} \bigg\{(1-\delta) \sum_{t=0}^{+\infty} \delta^t  x_t + c_1 (\delta)\bigg\}$$ and $$I(x)=\min_{\delta \in [0,1)} \bigg\{(1-\delta) \sum_{t=0}^{+\infty} \delta^t  x_t + c_2(\delta)\bigg\}.$$
For $i\in \{1,2\}$, define $\bar{c}_i: D \to [0,\infty]$ as $\bar{c}_i ((1-\delta)(1,\delta,\delta^2,\dots)) = c_i(\delta)$ for every $\delta \in [0,1)$, and $\bar{c}_i (p) = \infty$ for every $p \in D \setminus F$. Then we can prove that $\overline{\co} (\bar{c}_1) = \overline{\co} (\bar{c}_2)$, where $\overline{\co} (\bar{c}_i)$ is the lower semicontinuous convex envelope of $\bar{c}_i$ for $i\in \{1,2\}$.

Indeed, define $R_i=-\overline{\co} (\bar{c}_i)$, and define $\bar{R}_i :ba(\mathbb{N})\rightarrow \mathbb{R}\cup \{-\infty\}$ as follows: $\bar{R}_i(\mu)=R_i(\mu)$ for every $\mu \in D$, and $\bar{R}_i(\mu)= -\infty$ for every $\mu \not\in D$. Clearly, $\bar{R}_i$ is a weakly$^{\star}$ upper semicontinuous and concave function, and we have $$I(x) = \min_{\mu\in ba(\mathbb{N})}\{\langle x,\mu\rangle - \bar{R}_i(\mu)\}= \bar{R}_i^{\star}(x).$$
By Biconjugation theorem, $\bar{R}_i^{\star \star} =\bar{R}_i$, which implies that, for every $\mu \in D$, $$I^{D}(\mu) = \inf_{x \in l_{\infty}}\{\langle x,\mu\rangle - {I}(x)\}= \inf_{x \in l_{\infty}}\{\langle x,\mu\rangle - {\bar{R}_i^{\star}}(x)\}= \bar{R}_i^{\star \star}(\mu) =\bar{R}_i(\mu).$$ Thus, we can conclude that $\overline{\co} (\bar{c}_1) = \overline{\co} (\bar{c}_2)$.


\section{Proofs for Section \ref{section3}}

\subsection{Proof of Theorem \ref{corollary1}}\label{proof-coro1}

We use $l_1$ to denote the space of absolutely summable sequences, endowed with the norm $\| x \|_1 = \sum_{t=0}^{\infty} |x_t|$. The norm dual of $(l_1,\|.\|_1)$ is $l_{\infty}$ under the duality $\langle x,y \rangle = \sum_{t=0}^{\infty} x_t y_t$. For any $\mu \in ca(\mathbb{N})$, it can be uniquely represented by an element $x \in l_1$ where $x_t = \mu(\{t\})$. Conversely, given $x \in l_1$, it uniquely defines a measure $\mu \in ca(\mathbb{N})$ such that $\mu(A) = \sum_{t \in A} x_t$ for every $A \subseteq \mathbb{N}$. In the proof, we will utilize the notation $p=(p_0,p_1,p_2,\dots) \in l_1$ to represent a point in $ca(\mathbb{N})$.



\begin{part1}
Assume that $\succsim$ satisfies Axioms \ref{axiom1}-\ref{axiom8} and IDIS. By Proposition \ref{pro1}, there exists a unique weakly-increasing, concave, $1$-lipschitz, normalized, and translation invariant constant equivalent $I:l_{\infty}\rightarrow \mathbb{R}$ that represents $\succsim$. 

Let $F= \{(1-\delta)(1,\delta,\delta^2,\dots): \delta \in [0,1)\}$. From Lemma \ref{lem2}, in order to prove $I(x)=\min_{p \in F} \{\langle x,p \rangle +c(p)\}$ for some function $c:F \to \mathbb{R}\cup \{\infty\}$, we only have to prove that $\partial^{F} I(x) \neq \emptyset$ for every $x \in l^{\infty}$.

Let $x\in l_{\infty}$. We have $\partial ^{F}I(x) = \partial I(x) \cap F$, where $$\partial I(x)= \{\mu \in ba(\mathbb{N}): I(y)-I(x)\leq \langle y-x,\mu \rangle \text{ for every $y \in l_{\infty}$} \}$$ is the superdifferential of $I$ at $x$. Since $x$ is in the interior of the domain of the proper\footnote{$I$ is proper by definition because it does not take infinite values.} concave continuous function $I$, we have $\partial I(x) \neq \emptyset$ (see Theorem 7.12 in \cite{AB2006}). 

\begin{lemma}
$\partial I (x) \subseteq \Delta^{\sigma}$.
\end{lemma}
\begin{proof}
We first prove that $\partial I(x) \subseteq \Delta$. Let $p \in \partial I(x)$, then we have $I(x+1)-I(x)\leq \langle 1,p \rangle$ and $I(x-1)-I(x)\leq \langle -1,p \rangle$. Since $I$ is translation invariant, $I(x+1)-I(x)=1$ and $I(x-1)-I(x)=-1$. Thus, we get that $\langle 1,p \rangle=1$, i.e., $p$ has a total mass equal to 1. Let $(y,z) \in l_{\infty}$ such that $y\geq z$. Since $I$ is monotone, $0 \leq I(y-z+x)-I(x)\leq \langle y-z,p \rangle$, so that $p$ is monotone. Thus, we get that $p \in \Delta$.

Now, we will show that $\partial I (x) \subseteq \Delta^{\sigma}$. Let $p \in \partial I(x)$, let $(E_n)_{n\in \mathbb{N}}$ be a sequence of subsets of $\mathbb{N}$ with $(E_n)\downarrow \emptyset$, and let $k=\inf\{x_n-1: n \in \mathbb{N}\}$. Then for every $n\in \mathbb{N}$, by definition, we have 
\begin{align}
I(k E_{n}x)-I(x)\leq I(x-\mathbb{I}_{E_n})-I(x) \leq -\langle \mathbb{I}_{E_n},p\rangle. \label{1bis}
\end{align}
Let us prove that $\lim_{n\rightarrow \infty} I(k E_{n}x) = I(x)$. For every $m\in \mathbb{N}\setminus \{0\}$, $x \succ I(x) - \frac{1}{m}$. Consequently, from Axiom \ref{axiom8}, there exists $k_m \in \mathbb{N}$ such that $k E_{k_m}x \succ I(x) - \frac{1}{m}$, or $I(k E_{k_m}x)>I(x)-\frac{1}{m}$. Notice that $k E_{n}x \leq k E_{n+1}x < x$ for every $n\in \mathbb{N}$ since $k<\inf_{n\in \mathbb{N}} x_n$, it then follows that $I(x) \geq \lim_{n\rightarrow \infty} I(k E_{n}x) \geq I(x)-\frac{1}{m}$ for every $m \in \mathbb{N} \setminus \{0\}$. Thus, $\lim_{n\rightarrow \infty} I(k E_{n}x) = I(x)$. From (\ref{1bis}), we obtain $\lim_{n\rightarrow \infty} \langle \mathbb{I}_{E_n},p\rangle = 0$, which proves that $p \in \Delta^{\sigma}$.
\end{proof}

If $(1,0,0,\dots) \in \partial I(x)$, then it is evident that $\partial I(x) \cap F \neq \emptyset$. Now, we consider the case where $(1,0,0,\dots) \notin \partial I(x)$. In particular, the function $s: \partial I(x) \to \Delta^{\sigma}$ defined by $$s(p)= \frac{1}{1-p_0}(p_1,p_2,\dots)$$ is well defined. The forthcoming lemma plays a central role in our proof:

\begin{lemma}\label{centrallemma}
Assume that $(1,0,0,\dots) \notin \partial I(x)$. If $p \in \partial I(x)$, then $s(p) \in \partial I(x).$
\end{lemma}

\begin{proof}
Let $p \in \partial I(x)$. Observe that the condition $s(p) \in \partial I(x)$ is equivalent to, for every $y \in l_{\infty}$,
\begin{align*}
&I(y)-I(x) \leq \langle y-x,s(p) \rangle \\
\Longleftrightarrow \ &0  \leq \langle y-x+ I(x)- I(y),s(p) \rangle.
\end{align*}
The equivalence above follows from $I(y)-I(x) = \langle I(y)- I(x), s(p) \rangle$ (since $s(p)\in \Delta^{\sigma}$).

Define $x'=x-I(x)$ and $y'=y-I(y)$. The inequality above is equivalent to, for every $y' \in l_{\infty}$ with $y' \sim 0$,
\begin{align*}
&0 \leq \langle y'-x' ,s(p) \rangle \\
\Longleftrightarrow \ &0 \leq \langle y'-x' ,(p_1,p_2,\dots) \rangle \\
\Longleftrightarrow \ &0 \leq \langle (0,y'-x'),p \rangle.
\end{align*}

Now, since $p \in \partial I(x)$, we have $$I((0,y'-x')+x)-I(x) \leq \langle (0,y'-x'),p \rangle $$ for every $y \in l_{\infty}$. To conclude the proof of the lemma, it is sufficient to demonstrate that $$I((0,y'-x')+x) \geq I(x).$$ 
Define $d=y'-x'$. We have $$I(x+d) = I(x+y'-x')= I(x+y-I(y) - x+I(x)) = I(x).$$ Thus, from IDIS, we can establish that $I(x+(0,y'-x')) = I(x+(0,d)) \geq I(x)$, implying that $s(p) \in \partial I(x)$.
\end{proof}

Therefore, we can state that $s: \partial I(x) \to \partial I(x)$. It is evident that $s$ is continuous. Furthermore, $\partial I(x)$ is a convex and compact subset of $\Delta$. Utilizing the Schauder–Tychonoff fixed-point theorem (see Appendix \ref{schauder}), we can deduce that there exists $p \in \partial I(x)$ such that $s(p)=p$. With a straightforward computation, we can derive that $$p = (1-\delta)(1,\delta,\delta^2,\dots),$$ for some $\delta \in (0,1)$, i.e., $p \in F$. Therefore, we can conclude that $\partial I(x) \cap F \neq \emptyset$. From Lemma \ref{lem2}, we get that
\begin{align}\label{vardis}
I(x)=\min_{p \in F}\{\langle x,p\rangle + c'(p)\}.
\end{align}
Moreover, for every $p\in F$, the cost function $c':F \to \mathbb{R}\cup \{\infty\}$ is defined by
\begin{align}\label{defcc}
c'(p)= \sup_{x\in l_{\infty}}\{I(x)- \langle x,p\rangle \}.
\end{align}

Observe that, for any $p\in F, c'(p) \geq I(1) - \langle 1,p\rangle = 0$, so $c'$ is non-negative. Moreover, the normalization of $I$ verifies that $c'$ is grounded. In fact, as $I(1) = 1$, we can deduce from Equation (\ref{vardis}) that $\min_{p \in F} c'(p) = 0$.



Define the function $f: [0,1) \to \Delta^{\sigma}$ as $f(\delta) = (1-\delta)(1,\delta,\delta^2,\dots)$ for all $\delta \in [0,1)$. In addition, define the function $c: [0,1] \to \mathbb{R} \cup \{\infty\}$ in the following manner: $c(\delta) = c'(f(\delta))$ for $\delta \in [0,1)$, and $c(\delta) = \infty$ if $\delta = 1$. With these definitions in place, we can express: $$I(x)=\min_{\delta \in [0,1]} \bigg\{(1-\delta) \sum_{t=0}^{+\infty} \delta^t  x_t + c(\delta)\bigg\}.$$

Given that $c'$ is non-negative and grounded, it follows that $c$ is also non-negative and grounded. We will now demonstrate that $c$ is lower semicontinuous. Let $\alpha\in \mathbb{R}$. We have 
\begin{align}\label{trinhlon}
    \{\delta \in [0,1]: c(\delta) \leq \alpha\} = \{\delta \in [0,1): c(\delta) \leq \alpha\} = \{\delta \in [0,1): c'(f(\delta)) \leq \alpha\}.
\end{align}
If $c'$ is defined for all $p\in \Delta$ according to Equation (\ref{defcc}), then, for every $t \geq 0$, it follows that $\{p \in \Delta: c'(p) \leq t\} \subseteq \Delta^{\sigma}$ (as established in Theorem 2 in \cite{BBW2024}). This implies that $c'(p) = \infty$ for all $p\in \Delta \setminus \Delta^\sigma$.

Denote $\overline{F}$ the closure of $F$. We can show that $\overline{F} \setminus F \subseteq \Delta \setminus \Delta^{\sigma}$. To establish this, let us consider the contrary assumption that there exists $p\in \overline{F} \setminus F$ such that $p\notin \Delta \setminus \Delta^{\sigma}$, which implies that $p \in \Delta^{\sigma}$. Because $p\in \overline{F} \setminus F$, there exists a sequence $(\delta_n)_{n\in \mathbb{N}}$ in $[0,1)$ for which the sequence $(f(\delta_n))_{n\in \mathbb{N}}$ converges to $p$. We can assume that $(\delta_n)_{n\in \mathbb{N}}$ converges to a point $\bar{\delta} \in [0,1]$. If $\bar{\delta} = 1$, then $\langle \mathbf{e}_t, p\rangle = \lim_{n\rightarrow \infty}\langle \mathbf{e}_t,f(\delta_n) \rangle = 0$, where $\mathbf{e}_t$ is the sequence whose $t^{th}-$term is one and all other terms are zero. Thus, $p \notin \Delta^{\sigma}$, which is a contradiction. If $\bar{\delta} \in [0,1)$, then the sequence $(f(\delta_n))_{n\in \mathbb{N}}$ converges in norm to $f(\bar{\delta}) \in F$. This can be deduced from the fact that $f:[0,1) \to l_1$ is norm continuous, as established in Lemma 8 in \cite{CE2018}. Consequently, $(f(\delta_n))_{n\in \mathbb{N}}$ converges pointwise to both $p$ and $f(\bar{\delta})$, leading to the conclusion that $p = f(\bar{\delta}) \in F$, which is a contradiction. Thus, we establish that $\overline{F} \setminus F \subseteq \Delta \setminus \Delta^{\sigma}$.

Thus, we have $$\{p \in \overline{F}: c'(p) \leq \alpha\} = \{p \in F: c'(p) \leq \alpha\}.$$ Clearly, $c'$ defined on $\Delta$ is lower semicontinuous since it is a pointwise supremum of a family of lower semicontinuous functions.  Hence, $\{p \in \overline{F}: c'(p) \leq \alpha\}$ is closed (in $\Delta$), which implies that $\{p \in F: c'(p) \leq \alpha\}$ is closed (in $\Delta$). Let $E = \{\delta \in [0,1): c'(f(\delta)) \leq \alpha\}$. We can conclude that the set $\{f(\delta): \delta \in E\}$ is closed (in $\Delta$). Furthermore, given that $\Delta$ is a compact set, it follows that ${f(\delta): \delta \in E}$ is also compact.

Assume that $E$ is open (in $[0,1]$). There is a sequence $(\delta_n)_{n\in \mathbb{N}}$ in $E$ that converges to $\bar{\delta} \in [0,1]$, but $\bar{\delta} \notin E$. As $(f(\delta_n))_{n\in \mathbb{N}}$ is contained within the compact set $\{f(\delta): \delta \in E\}$, there exists a subsequence, which we will continue to denote as $(f(\delta_n))_{n\in \mathbb{N}}$, that converges to $f(\delta')$ for some $\delta' \in E$. This implies that $(\delta_n)_{n\in \mathbb{N}}$ converges to $\delta'$. Consequently, we must have $\bar{\delta} = \delta'$, which is a contradiction. Therefore, $E$ is closed (in $[0,1]$). From (\ref{trinhlon}), we can deduce that $c$ is lower semicontinuous.

\end{part1}

\begin{part2}

For a given preference relation $\succsim$, suppose that  there exists a grounded and lower semicontinuous function $c:[0,1] \to [0,\infty]$, where $c(1) = \infty$, such that the constant equivalent representing $\succsim$ can be expressed as:
\begin{align}\label{eqdefIx}
I(x)=\min_{\delta \in [0,1]} \bigg\{(1-\delta) \sum_{t=0}^{+\infty} \delta^t  x_t + c(\delta)\bigg\}.
\end{align}

The proof that $\succsim$ satisfies Axioms \ref{axiom1} to \ref{axiom8} is left to the reader. We will now demonstrate that $\succsim$ satisfies IDIS. For every $x\in l_{\infty}$ and $\delta \in [0,1)$, let $D_{\delta} (x) = (1-\delta) \sum_{t=0}^{+\infty} \delta^t  x_t$. Consider $(x, d) \in l_{\infty}$ such that $x + d \succsim x$. Our objective is to show that $x + (0, d) \succsim x$. By (\ref{eqdefIx}), there is $\delta \in [0,1]$ such that $$I(x+(0,d)) = D_{\delta}(x+(0,d)) +  c(\delta).$$ 

Observing that $D_{\delta}((0, d))=\delta D_{\delta}(d)$, we obtain the following relationship:
\begin{align*}
D_{\delta}(x+(0,d)) +  c(\delta) &= D_{\delta}(x)+\delta D_{\delta}(d) +  c(\delta)\\
&= (1-\delta)[D_{\delta}(x) + c(\delta)] + \delta[D_{\delta}(x+d) + c(\delta)] \\
&\geq (1-\delta)I(x) + \delta I(x+d) \\
&\geq (1-\delta)I(x) + \delta I(x) = I(x).
\end{align*}
The third line can be deduced from (\ref{eqdefIx}), and the last line is a direct consequence of the assumption that $x + d \succsim x$. Therefore, we have established that $I(x + (0, d)) \geq I(x)$, which implies that $x + (0, d) \succsim x$. As a result, we can conclude that $\succsim$ satisfies IDIS.
\end{part2}


\subsection{Proof of Proposition \ref{corollary2}}\label{proof-coro2}

We will solely demonstrate the ``only-if" part, as the ``if" part is straightforward. Let us assume that $\succsim$ satisfies Axioms \ref{axiom1}-\ref{axiom6} and IDIS. Define $F= \{(1-\delta)(1,\delta,\delta^2,\dots): \delta \in [0,1)\}$. Introduce the function $f: [0,1) \to \Delta^{\sigma}$ as $f(\delta) = (1-\delta)(1,\delta,\delta^2,\dots)$ for all $\delta \in [0,1)$. According to the proof of Theorem \ref{corollary1}, the constant equivalent representing $\succsim$ can be expressed as follows:
$$I(x)=\min_{\delta \in [0,1]} \bigg\{(1-\delta) \sum_{t=0}^{+\infty} \delta^t  x_t + c(\delta)\bigg\},$$ where $c: [0,1] \to [0,\infty]$ is a grounded and lower semicontinuous function defined as:
\begin{align*}
    c(\delta)= 
\begin{cases}
\sup_{x\in l_{\infty}} \bigg\{I(x) - (1-\delta) \sum_{t=0}^{+\infty} \delta^t  x_t \bigg\} &\text{ if $\delta \in [0,1)$}\\
\infty &\text{ if $\delta =1$}
\end{cases}.
\end{align*}

Let $\alpha \in \mathbb{R}$ such that $\alpha >0$. Let us prove that $I(\alpha x)=\alpha I(x)$. We will employ a proof by contradiction. Assume for the sake of contradiction that $I(\alpha x) > \alpha I(x)$. This implies that $\alpha x\succ \alpha I(x)$. Using Axiom \ref{axiom7}, we conclude that $x\succ I(x)$, leading to a contradiction. Now, consider the opposite case, where we assume that $I(\alpha x) < \alpha I(x)$. In this scenario, we find that $I(x) \succ x$, again resulting in a contradiction. Hence, we must conclude that neither of these contradictions holds, and the only consistent conclusion is that $I(\alpha x) = \alpha I(x)$, i.e., $I$ exhibits positive homogeneity.

Consequently, this implies that, for any given $\delta \in [0,1)$, the function $g_{\delta}:l_{\infty} \to \mathbb{R}$ defined as $$g_{\delta}(x) = I(x)- (1-\delta) \sum_{t=0}^{+\infty} \delta^t x_t$$ for all $x\in l_{\infty}$ also exhibits positive homogeneity. If there exists $x\in l_{\infty}$ such that $g_{\delta}(x) > 0$, then $g_{\delta}(n x) = n g_{\delta}(x)$ diverges to infinity as $n$ approaches infinity. Consequently, we have $c(\delta) = \sup_{x\in l_{\infty}} g_{\delta}(x) = \infty$. On the other hand, if $g_{\delta}(x) \leq 0$ for all $x\in l_{\infty}$, it's evident that $c(\delta) = 0$. Therefore, we conclude that $c(\delta)$ can only take on one of two values: either $c(\delta) = 0$ or $c(\delta) = \infty$. Let $E = \{\delta \in [0,1]:c(\delta)=0\}$. This set is non-empty and closed since $c$ is both grounded and lower semicontinuous. Moreover, $E \subseteq [0,1)$ since $c(1)=\infty$. We can then express $I(x)$ as: $$I(x)=\min_{\delta \in E} \bigg\{(1-\delta) \sum_{t=0}^{+\infty} \delta^t  x_t\bigg\}.$$
Finally, for the uniqueness of $E$, we refer to Theorem 3 in \cite{CE2018}.

\subsection{Proof of Proposition \ref{corollary-exponential}}\label{proof-coroexpo}

We only prove the ``only if" part, the other direction being routine. Assume that $\succsim$ is a binary relation on $l_{\infty}$ satisfying Axioms \ref{axiom1}-\ref{axiom7}, Strong monotonicity, and IDIS. Based on Proposition \ref{corollary2} and Remark \ref{rek1}, we can express the unique constant equivalent $I: l_{\infty} \to \mathbb{R}$ of $\succsim$ as follows:
\begin{align*}
I(x)=\min_{\delta \in E} \bigg\{(1-\delta) \sum_{t=0}^{+\infty} \delta^t  x_t\bigg\}
\end{align*}
for some closed set $E \subseteq (0,1)$.

Evidently, as $\succsim$ adheres to Axiom \ref{axiom6}, the function $I$ exhibits positive homogeneity. From Axiom \ref{axiom7}, for all $x,y\in l_{\infty}$, we have $$x+y \sim I(x)+y \sim I(x) + I(y).$$ Thus, $I(x+y)= I(x)+I(y)$, i.e., the function $I$ is additive. In particular, take $y = -x$, we get that $I(x) = - I(-x)$. Let $\lambda \in \mathbb{R}$ such that $\lambda < 0$. We have $$I(\lambda x) = I(-\lambda (-x)) = -\lambda I(-x) = \lambda I(x).$$ Hence, the function $I$ is linear on $l_{\infty}$. This implies that $E$ must be a singleton, which ends the proof of the proposition.

\subsection{Proof of Proposition \ref{propaggre}}

We establish the proof for the implication of $(i)$ to $(ii)$, while the converse direction is straightforward. Let us assume that the social preference is represented by the variational discounting criterion with the cost function $c$, and that Unanimity is satisfied. We will proceed by contradiction. Suppose that there exists a discount factor $\delta \notin D$ such that $c(\delta) < \infty$. For a given $\alpha > 0$, consider the sequence $x \in l_{\infty}$ defined as follows: $x_0 = \alpha \delta^2, x_1 = \alpha (\delta^2 - 2\delta)$, and $x_t = \alpha (1 - \delta)^2$ for every $t > 1$. With a straightforward computation, we can derive the following result: 
\begin{align*}
    (1-\delta) \sum_{t} \delta^t x_t = 0 \leq \alpha (\delta'-\delta)^2 = (1-\delta') \sum_{t} (\delta')^t x_t.
\end{align*}

Since $D$ is closed and $\delta \notin D$, there exists $\alpha > 0$ such that $c(\delta) < \alpha \min_{\delta' \in D} (\delta'-\delta)^2$. Let $\alpha \min_{\delta' \in D} (\delta'-\delta)^2 =\theta$. We have $(1-\delta') \sum_{t} (\delta')^t x_t \geq \theta$ for all $\delta \in D$. Hence, due to Unanimity, we establish that $x \succsim \theta$. However, it is clear that $\theta > (1 - \delta) \sum_{t} \delta^t x_t + c(\delta)$. Therefore, we have $\theta \succ x$, leading to a contradiction. Consequently, through this contradiction, we can deduce that for all $\delta \notin D$, it must hold that $c(\delta) = \infty$, implying that the social preference can be expressed by $$I(x)=\min_{\delta \in D} \bigg\{(1-\delta) \sum_{t=0}^{+\infty} \delta^t  x_t + c(\delta)\bigg\}.$$

\section{Proofs for Section \ref{section4}}

\subsection{Proof of Lemma \ref{lemdefadj}}  \label{proof-dual}

Let $T: l_{\infty} \to l_\infty$ be a continuous positive linear function, and let $\mu \in ba(\mathbb{N}$. We aim to establish the unique definition of $T^*(\mu)$. Given the continuous linearity of $T$, it is clear that the functional $f_{\mu}: l_\infty \rightarrow \mathbb{R}$ defined by $f_{\mu}(x) = \langle T(x), \mu \rangle$ for every $x \in l_\infty$ is a continuous linear function. As the topological dual of $l_\infty$ is lattice isometric to $ba(\mathbb{N)}$, there exists a unique ${\mu}^{*}\in ba(\mathbb{N})$ such that $\langle T(x), \mu \rangle = \langle x, \mu^{*} \rangle$ for every $x \in l_\infty$. Consequently, we define $T^*(\mu) = \mu^*$, which is evidently a uniquely defined transformation.

Therefore, the function $T^*: ba(\mathbb{N}) \rightarrow ba(\mathbb{N})$ is well defined. Now, we need to demonstrate that $T^*$ is a continuous positive linear function. To establish the linearity of $T^*$, consider $\mu, \mu' \in ba(\mathbb{N})$. We have $$\langle T(x),\mu +\mu' \rangle = \langle T(x),\mu \rangle+\langle T(x),\mu' \rangle \text{ for every $x\in l_{\infty}$}.$$ This is equivalent to $\langle x, T^*(\mu+\mu') \rangle= \langle x, T^*(\mu) \rangle+\langle x, T^*(\mu') \rangle$ for every $x \in l_{\infty}$. Thus, we conclude that $T^*(\mu + \mu') = T^*(\mu) + T^*(\mu')$, demonstrating that $T^*$ is additive. Now, let $\mu \in ba(\mathbb{N})$ and $\lambda \in \mathbb{R}$. We have $$\langle x,T^*(\lambda\mu) \rangle = \lambda \langle T(x),\mu \rangle \text{ for every $x\in l_{\infty}$},$$ which is equivalent to $\langle x,T^*(\lambda\mu) \rangle = \lambda \langle x,T^*(\mu) \rangle$ for every $x \in l_{\infty}$. This implies that $T^*(\lambda \mu) = \lambda T^*(\mu)$. Therefore, $T^*$ is a linear function.

Next, we will prove that $T^*$ is continuous. Let $(\mu_d)_{d\in D}$ be a net in $ba(\mathbb{N})$ such that $\mu_d$ converges to $\mu \in ba(\mathbb{N})$. By definition, for every $d \in D$, $\langle x, T^*(\mu_d)\rangle = \langle T(x), \mu_d \rangle$. As $(\langle T(x), \mu_d \rangle)_{d}$ converges to $\langle T(x), \mu \rangle = \langle x, T^*(\mu) \rangle$, it follows that $(\langle x, T^*(\mu_d) \rangle)_{d}$ also converges to $\langle x, T^*(\mu) \rangle$. This convergence holds for every $x \in l_{\infty}$. Thus, by the definition of weak$^{\star}$ convergence, $T^*(\mu_d)$ converges to $T^*(\mu)$. Hence, $T^*$ is continuous.

Finally, we establish that $T^*: ba(\mathbb{N}) \to ba(\mathbb{N})$ is a positive operator. Let $\mu \in ba(\mathbb{N})$ and $x \in l_{\infty}$, both satisfying $\mu \geq 0$ and $x \geq 0$. In this context, we have $\langle x, T^*(\mu) \rangle = \langle T(x), \mu \rangle \geq 0$ because $T(x) \geq 0$. Therefore, $T^*(\mu): l_{\infty} \to \mathbb{R}$ is positive, which consequently implies that $T^*$ is a positive operator.

\subsection{Proof of Theorem \ref{theorem1}}
\begin{part1}
Assume that $\succsim$ satisfies Axioms \ref{axiom1}-\ref{axiom5} and I$T$IS. Let $I:l_{\infty}\rightarrow \mathbb{R}$ be the unique weakly-increasing, concave, $1$-lipschitz, normalized, and translation invariant constant equivalent given in Proposition \ref{pro1}. According to Lemma \ref{lem2}, the proof of $I(x) = \min_{p \in \Delta (T^*)} \{\langle x,p \rangle + c(p)\}$ for a function $c: \Delta (T^*) \to \mathbb{R} \cup \{\infty\}$ only necessitates demonstrating that $\partial^{\Delta (T^*)} I(x) \neq \emptyset$ for all $x \in l^{\infty}$. 

Let $x\in l_{\infty}$. We have $\partial ^{\Delta (T^*)}I(x) = \partial I(x) \cap \Delta (T^*)$. Clearly, $\partial I(x)$ is non-empty and is a subset of $\Delta$. If there exists $p \in \partial I(x)$ such that $T^*(p) = 0$, then $p\in \Delta (T^*)$. In addition, since $\partial I(x) \subseteq \Delta$, $p\in \Delta$, implying that $\partial^{\Delta (T^*)} I(x) \neq \emptyset$. Now, we consider the case where $T^*(p) \neq 0$ for every $p \in \partial I(x)$. In particular, we can define the function $s: \partial I(x) \to ba(\mathbb{N})$ as $$s(p)=\frac{T^*(p)}{\langle 1, T^*(p) \rangle} \text{ for every $p\in \partial I(x)$}.$$

Observe that, for every $p \in \partial I(x) \subseteq \Delta$, $s(p) \in \Delta$, so $s$ maps from $\partial I(x)$ to $\Delta$. By following a similar argument as in the proof of Lemma \ref{centrallemma}, we can demonstrate that if $p \in \partial I(x)$, then $s(p) \in \partial I(x)$. Hence, $s$ constitutes a self-mapping function from $\partial I(x)$ to itself. Moreover, $s$ is continuous, and $\partial I(x)$ is convex and compact. By Schauder–Tychonoff fixed-point theorem, there exists $p \in \partial I(x) $ such that $s(p)=p$. Notably, $p$ is an element of $\Delta$, which also implies that $p \in \Delta (T^*)$. Consequently, we establish that $\partial^{\Delta (T^*)} I(x) \neq \emptyset$. This result, in accordance with Lemma \ref{lem2}, yields
\begin{align}\label{fi}
I(x)=\min_{p \in \Delta (T^*)}\{\langle x,p\rangle + c(p)\},
\end{align}
where, for every $p\in \Delta (T^*)$,
\begin{align*}
c(p)= \sup_{x\in l_{\infty}}\{I(x)- \langle x,p\rangle \}.
\end{align*}

Consider that, for any $p \in \Delta (T^*)$, it holds that $c(p) \geq I(1) - \langle 1, p\rangle = 0$. This observation confirms that $c$ is non-negative. Clearly, $c$ is lower semicontinuous. Additionally, the normalization of $I$ implies that $c$ is grounded. To illustrate, since $I(1) = 1$, we can derive from Equation (\ref{fi}) that $\min_{p \in \Delta (T^*)} c(p) = 0$. 
\end{part1}

\begin{part2}
Suppose that there exists a grounded and lower semicontinuous function $c: \Delta (T^*) \rightarrow [0,\infty]$ such that the unique constant equivalent $I:l_{\infty}\rightarrow \mathbb{R}$ representing $\succsim$ can be written as follows: 
\begin{align*}
I(x)=\min_{p \in \Delta (T^*)}\{\langle x,p\rangle + c(p)\}.
\end{align*}

We will only prove that $\succsim$ satisfies I$T$IS. Let us assume that $\langle 1, T^*(p) \rangle \leq 1$ for every $p \in \Delta (T^*)$. Let $x,d \in l_{\infty}$ such that $x + d \succsim x$. We need to show that $x + T(d) \succsim x$. By definition, there exists $p \in \Delta (T^*)$ such that $I(x + T(d)) = \langle x + T(d), p \rangle + c(p)$. According to the definition of $\Delta (T^*)$, we have two cases: either $T^*(p) = 0$ or $T^*(p) = \lambda p$ for some $\lambda > 0$.

If $T^*(p) = 0$, then we can proceed as follows: 
\begin{align*}
\langle x+T(d),p \rangle + c(p) &= \langle x,p \rangle + \langle T(d), p \rangle + c(p)\\
&= \langle x,p \rangle + \langle d, T^*(p) \rangle + c(p) \\
&= \langle x,p \rangle + c(p) \geq I(x).
\end{align*}
This implies that $x+T(d) \succsim x$.

Consider the case where $T^*(p) = \lambda p$ for some $\lambda > 0$. We first observe that $1 \geq \langle 1, T^*(p) \rangle = \langle 1, \lambda p \rangle = \lambda$. Then we have
\begin{align*}
\langle x+T(d),p \rangle + c(p) &= \langle x,p \rangle + \langle d, T^*(p) \rangle + c(p)\\
&= \langle x,p \rangle + \lambda \langle d,  p \rangle + c(p)\\
&= (1-\lambda)[\langle x,p \rangle + c(p)] + \lambda[\langle x+d,p \rangle + c(p)] \\
&\geq (1-\lambda)I(x) + \lambda I(x) = I(x).
\end{align*}
The last inequality is due to the fact that $1 \geq \lambda$, the definition of $I$, and the assumption $x+d \succsim x$. Thus, we can conclude that $x+T(d) \succsim x$.

Therefore, $\succsim$ satisfies I$T$IS.
\end{part2}

\subsection{Proof of Proposition \ref{promaxmin}}

We only prove the first implication, the other direction being routine. Assume that $\succsim$ is a binary relation on $l_{\infty}$ satisfying Axioms \ref{axiom1}-\ref{axiom5}, Axiom \ref{axiom6} and I$T$IS. From the proof of Theorem \ref{theorem1}, the unique constant equivalent $I:l_{\infty}\to \mathbb{R}$ of $\succsim$ can be represented by 
\begin{align*}
I(x)=\min_{p \in \Delta(T^*)}\{\langle x,p\rangle + c(p)\},
\end{align*}
where, for every $p\in \Delta(T^*)$,
\begin{align*}
c(p)= \sup_{x\in l_{\infty}}\{I(x)- \langle x,p\rangle\}.
\end{align*}

Since $\succsim$ satisfies Axiom \ref{axiom6}, $I$ is positively homogeneous. This implies that, for any given $p \in \Delta(T^*)$, the function $g_{\delta}:l_{\infty} \to \mathbb{R}$ defined as $g_{\delta}(x) = I(x)- \langle x,p\rangle$ for all $x\in l_{\infty}$ also exhibits positive homogeneity. Thus, either $c(p)= 0$ or $c(p)= \infty$. Let $D = \{p\in F:c(p)=0\}$, which is non-empty since $c$ is grounded. Moreover, since $\Delta(T^*)$ is closed (see Remark \ref{remark1}) and $c$ is lower semicontinuous on $\Delta$, $D$ is closed. Finally, it is clear that we can write 
\begin{align*}
I(x)=\min_{p \in D} \langle x,p\rangle,
\end{align*} which finishes the proof.

\subsection{Proof of Theorem \ref{theorem2.3}}

We only proof the first part of the theorem. Assume that $\succsim$ satisfies Axioms \ref{axiom1}-\ref{axiom5}, Axiom \ref{axiom6}, and I$T_i$IS for every $i\in K$. It then follows from Proposition \ref{promaxmin} that, for every $i\in K$, $\succsim$ can be represented by 
\begin{align*}
I(x) = \min_{p\in D_i} \langle x, p\rangle,
\end{align*}
where $D_i$ is a closed subset of $F_i = \{p \in \Delta: \text{$p$ is an eigenvector of $T_i^{*}$}\}$. Clearly, we must have $\overline{\co} D_i = \overline{\co} D_j = E$ for every $i,j \in K$.

A point $p \in E$ is an \textit{exposed point} of $E$ if there exists $x \in l_{\infty}$ such that $\langle x, p' \rangle > \langle x, p \rangle$ for every $p' \in E \setminus \{p\}$. We use ${\cal E}(E)$ to denote the set of all exposed points of $E$. It is well known that a weakly compact convex set in a Banach space is the closed convex hull of its exposed points.\footnote{See Corollary 5.18 in \cite{BL1998}, and notice that \textit{strongly exposed} points of convex subsets of Hausdorff topological vector spaces are automatically exposed points.} Thus, we can express $$I(x)= \min_{p\in E} \langle x, p\rangle = \inf_{p\in {\cal E}(E)} \langle x, p\rangle = \min_{p\in \overline{{\cal E}(E)}} \langle x, p\rangle.$$ 

Now, we show that ${\cal E}(E) \subseteq D_{i}$ for every $i\in K$. By contradiction, assume that there is $p\in {\cal E}(E)$ such that $p \notin D_{i}$. By definition, there exists $x \in l_{\infty}$ such that $\langle x, p' \rangle > \langle x, p \rangle$ for every $p' \in E \setminus \{p\}$. But we have $$I(x) = \min_{p'\in D_i} \langle x, p' \rangle > \langle x, p\rangle = \min_{p'\in E} \langle x, p'\rangle = I(x),$$ which is a contradiction. Thus, we get that ${\cal E}(E) \subseteq D_{i}$ for every $i\in K$, which implies that $\overline{{\cal E}(E)} \subseteq D_{i}$ since $D_i$ is closed. This in turn implies that $\overline{{\cal E}(E)} \subseteq \bigcap_{i\in K} D_{i} \subseteq \bigcap_{i\in K} F_{i}$, and by letting $D = \overline{{\cal E}(E)}$, the proof is finished.

\subsection{Proof of Theorem \ref{theorem2}}\label{Proof-theorem2}

We only prove the first implication, the converse implication being routine. Assume that $\succsim$ is a binary relation on $l_{\infty}$ satisfying Axioms \ref{axiom1}-\ref{axiom8} and I$T$IS. To apply Lemma \ref{lem2}, we will prove that $\partial^{F}I(x) \neq \emptyset$ for every $x \in l_{\infty}$, where $F=\{p \in \Delta^{\sigma}: \text{$p$ is an eigenvector of $T^*$} \}$.

Clearly, $\partial I(x)$ is non-empty subset of $\Delta^{\sigma}$. Now, using the same argument in the proof of Theorem \ref{theorem1} (substitute $\Delta^{\sigma}$ for $\Delta$), we can deduce that $\partial^{F}I(x) \neq \emptyset$ for every $x \in l_{\infty}$. From Lemma \ref{lem2}, we obtain
\begin{align*}
I(x)=\min_{p \in F}\bigg\{\sum_{n=0}^{\infty} p_n x_n  + c'(p)\bigg\},
\end{align*}
where, for every $p\in F$,
\begin{align}\label{extended}
c'(p)= \sup_{x\in l_{\infty}}\bigg\{I(x)- \sum_{n=0}^{\infty} p_n x_n\bigg\}.
\end{align}
We define $c:\overline{F} \to [0,\infty]$ by $c(p) = c'(p)$ for all $p\in F$ and $c(p) = \infty$ for all $p\in \overline{F} \setminus F$. Then we can write
\begin{align*}
I(x)=\min_{p \in F}\bigg\{\sum_{n=0}^{\infty} p_n x_n  + c(p)\bigg\}.
\end{align*}

Clearly, $c$ is non-negative and grounded. By definition, $c$ is $F^c$-infinite. Now, we will prove that $c$ is lower semicontinuous. If $c'$ given by Equation (\ref{extended}) is defined for all $p\in \Delta$, then $\{p \in \Delta: c'(p) \leq t\} \subseteq \Delta^{\sigma}$ for every $t \geq 0$ (see Theorem 2 in \cite{BBW2024}), which implies that $c'(p) = \infty$ for all $p\in \Delta \setminus \Delta^\sigma$. Define $F'=\{p \in \Delta: \text{$p$ is an eigenvector of $T^*$} \}$, which is closed by Lemma \ref{remark1}. We have $$\{p \in \overline{F}: c(p) \leq t\} = \{p \in F: c'(p) \leq t\} = \{p \in F': c'(p) \leq t\}.$$
Since $c'$ defined on $\Delta$ is lower semicontinuous, $\{p \in F': c'(p) \leq t\}$ is closed in $\Delta$, which implies that $\{p \in \overline{F}: c(p) \leq t\}$ is closed in $\Delta$. Thus, $c$ is lower semicontinuous.

\subsection{Proof of Proposition \ref{prop-maxmin2}}\label{Proof-prop-maxmin2}

We only prove the first implication, the converse implication being routine. Assume that $\succsim$ is a binary relation on $l_{\infty}$ satisfying Axioms \ref{axiom1}-\ref{axiom6} and I$T$IS. From Theorem \ref{theorem2}, there exists a grounded, $F^c$-infinite, and lower semicontinuous function $c: \overline{F} \rightarrow [0,\infty]$ such that the unique constant equivalent $I:l_{\infty}\rightarrow \mathbb{R}$ representing $\succsim$ can be written by
\begin{equation*} 
I(x)=\min_{p \in F}\bigg\{\sum_{n=0}^{\infty} p_n x_n  + c(p)\bigg\},
\end{equation*} 
where $F=\{p \in \Delta^{\sigma}: \text{$p$ is an eigenvector of $T^*$} \}$. Clearly, ISU implies that, for every $p\in \overline{F}$, either $c(p) = 0$ or $c(p)= \infty$. Consequently, we can write 
\begin{align*}
I(x)=\min_{p \in D} \sum_{n=0}^{\infty} p_n x_n, 
\end{align*}
where $D = \{p \in \overline{F}:c(p)=0\}$, which is non-empty because $c$ is grounded. Now, since $c$ is lower semicontinuous, $D$ is closed in $\Delta$.

\section{Proofs for Section \ref{sec:equal}}

\subsection{Proof of Theorem \ref{theoremequal}}

Define the delayed transformation $T_{de}:l_{\infty} \to l_{\infty}$ as $T_{de}(x)=(0,x)$ for every $x \in l_{\infty}$. The next result shows that $p\in \Delta (T^*_s)$ is either a geometric distribution or a Banach-Mazur limit.

\begin{lemma}\label{lemmageo}
$\Delta (T^*_s) = \mathcal{B} \cup \{(1-\delta)(1,\delta,\delta^2,\dots):\delta \in [0,1)\}$.
\end{lemma}
\begin{proof}

Let $p \in \mathcal{B} \cup \{(1-\delta)(1,\delta,\delta^2,\dots):\delta \in [0,1)\}$. Then either $p$ is a Banach-Mazur limit or $p$ is a geometric distribution. We first consider that $p$ is a Banach-Mazur limit. We have $\langle (0,x), p \rangle = \langle x, p \rangle$ for all $x \in l_{\infty}$. By definition, we can conclude that $p$ is the adjoint of $T_{de}$, which implies that $p$ is a normalized eigenvector of $T^*_s$. Now, if $p$ is a geometric distribution, then $p = (1-\delta)(1,\delta,\delta^2,\dots)$ for some $\delta \in [0,1)$. We have $\langle (0,x), p \rangle = \langle x, \delta p \rangle$. Thus, $\delta p$ is the adjoint of $T_{de}$, which implies that $p$ is a normalized eigenvector of $T^*_s$. We have shown that $\mathcal{B} \cup \{(1-\delta)(1,\delta,\delta^2,\dots):\delta \in [0,1)\} \subseteq \Delta (T^*_s)$.

Now, we will prove that $\Delta (T^*_s) \subseteq \mathcal{B} \cup \{(1-\delta)(1,\delta,\delta^2,\dots):\delta \in [0,1)\}$. Let $p$ be a normalized eigenvector of $T^*_s$ associated to an eigenvalue $\lambda$. By definition, we have $$\langle (0,x), p \rangle = \langle x, T^*_s(p)\rangle = \lambda \langle x, p \rangle \text{ for every $x\in l_{\infty}$}.$$ 
In particular, taking $x=1$, we easily get that $0\leq \lambda \leq 1$. 

For $\lambda=1$, since $p$ is a discount structure, the condition above is equivalent to 
\begin{align*}
    &\langle (0,x), p \rangle = \langle x, p \rangle \\
    \Longleftrightarrow \ &\langle (\theta,x+\theta), p \rangle = \langle x+\theta, p \rangle \text{ for any $\theta \in \mathbb{R}$.}
\end{align*}
Define $x' = (\theta,x+\theta)$, then $(x'_1,x'_2,\dots) = x+\theta$. Since the equation above holds for any $x\in l_{\infty}$ and $\theta \in \mathbb{R}$, we get that $\langle x', p \rangle=\langle (x'_1,x'_2,\dots), p \rangle$ for any $x' \in l_{\infty}$. Thus, $p$ is a Banach-Mazur limit. 

For $\lambda=0$, $p$ is only supported on the first component, i.e., $\langle x, p \rangle = x_0$ for every $x\in l_{\infty}$. In particular, when $x=(0,1,1,\dots)$, then we get that $p(\{1,2,\dots\})=0$, which implies that $p$ is identified to the discounted sequence $(1,0,\dots)$. 

For $\lambda \in (0,1)$, one can easily deduce that $p(\{n+1\}) = \lambda p(\{n\})$ for every $n \in \mathbb{N}$. This in turn implies that $p$ is countably additive (since $p(\{n,n+1,\dots\})$ converges to 0 as $n$ goes to infinity) and coincides with the geometric sequence of parameter $\lambda$.  

To summarize, we must have $\Delta (T^*_s) \subseteq \mathcal{B} \cup \{(1-\delta)(1,\delta,\delta^2,\dots):\delta \in [0,1)\}$. Therefore, $\Delta (T^*_s) = \mathcal{B} \cup \{(1-\delta)(1,\delta,\delta^2,\dots):\delta \in [0,1)\}$. 
\end{proof}

We have $1 \geq \langle (0,1,1,1,\dots), p \rangle = \langle 1, T^*_s (p) \rangle$ for all $p\in \Delta (T^*_s)$. Thus, it follows from Theorem \ref{theorem1} that a binary relation $\succsim$ satisfies Axioms \ref{axiom1}-\ref{axiom5} and IDIS if and only if there exists a grounded and lower semicontinuous function $c: \mathcal{D} \to [0,\infty]$ such that the constant equivalent of $\succsim$ can be expressed as
\begin{equation*} 
I(x)=\min_{p \in \mathcal{D}} \{\langle x,p \rangle +c(p)\},
\end{equation*} which finishes the proof.

\subsection{Proof of Proposition \ref{corollaryBanach}}\label{proof-coroBanach}

We will only proof the ``only if" part. Assuming that $\succsim$ satisfies Axioms \ref{axiom1}-\ref{axiom5}, IDIS, and Time invariance, according to Theorem \ref{theoremequal}, the constant equivalent $I$ of $\succsim$ can be expressed as follows:
\begin{equation*} 
I(x)=\min_{p \in \mathcal{D}} \{\langle x,p \rangle +c(p)\},
\end{equation*} where $\mathcal{D} = \Delta (T^*_s) \subseteq \mathcal{B} \cup \{(1-\delta)(1,\delta,\delta^2,\dots):\delta \in [0,1)\}$. Additionally, just like all the proofs of variational representations presented earlier, we can select $c: \mathcal{B} \to [0,\infty]$, defined by
\begin{align*}
c(p)= \sup_{x\in l_{\infty}}\{I(x)- \langle x,p\rangle\}.
\end{align*}

Now, we will show that $c(p) = \infty$ for all $p \in \{(1-\delta)(1,\delta,\delta^2,\dots):\delta \in [0,1)\}$. Let $p = (1-\delta)(1,\delta,\delta^2,\dots)$ for some $\delta \in [0,1)$. We first notice that Time invariance implies that, for any $x\in l_{\infty}$ and $\theta \in \mathbb{R}$, $(\theta,x_1,x_2,x_3,\dots) \sim x$. Thus, $(-n,0,0,0,\dots) \sim 0$ for all $n \in \mathbb{N}$. We have $$c(p) \geq I((-n,0,0,0,\dots)) + (1-\delta)n = (1-\delta)n \text{ for all $n\in \mathbb{N}$}.$$ Let $n$ goes to $\infty$, we can conclude that $c(p) = \infty$. Therefore, we can express $I$ as follows:
\begin{equation*} 
I(x)=\min_{p \in \mathcal{B}} \{\langle x,p \rangle +c(p)\}.
\end{equation*}

We remark that the set of Banach-Mazur limits is closed and convex. Clearly, $c$ is grounded, lower semicontinuous, and convex. Finally, the uniqueness of $c$ can be deduced following the same argument presented in Appendix \ref{appenunique}.

\subsection{Proof of Theorem \ref{corollary7}}

For any finite permutation $\sigma:\mathbb{N} \to \mathbb{N}$, define $C_{\sigma} = \{\mu \in ca(\mathbb{N}): \mu = \mu_{\sigma}\}$, where $\mu_{\sigma}=(\mu_{\sigma(0)},\mu_{\sigma(1)},\mu_{\sigma(2)},\dots)$. Define the finite permutation transformation with respect to $\sigma$ as $T_{\sigma}(x) = x_{\sigma}$. The following lemma characterizes eigenvectors of $T^*_{\sigma}$:

\begin{lemma}\label{lemmaper}
The set of eigenvectors of $T^*_{\sigma}$ is a subset of $C_{\sigma} \oplus pa(\mathbb{N})$, and all of them have eigenvalue 1.
\end{lemma}

\begin{proof}
Let $\mu$ be an eigenvector of $T^*_{\sigma}$. By definition, there is $\lambda \in \mathbb{R}$ such that $\lambda \langle x, \mu \rangle = \langle x_{\sigma}, \mu \rangle$ for all $x\in l_{\infty}$. In particular, if we set $x = 1$, we obtain $\lambda = 1$, implying that $\lambda = 1$ is the unique eigenvalue of $T^*_{\sigma}$. For every $x \in l_{\infty}$, we have

\begin{equation}\label{cite}
\langle x, \mu \rangle = \langle x_{\sigma}, \mu \rangle.
\end{equation}

Clearly, $T_{\sigma}$ satisfies two conditions in Lemma \ref{criteriu}, which implies that $T^*_{\sigma}$ is a YH transformation. By Lemma \ref{lemma3}, in order to characterize the set of eigenvectors $E_{1}(T^*_{\sigma})$, our task is reduced to identifying the set of countably additive eigenvectors $E_{1}^{ca}(T^*_{\sigma})$ and the set of purely finitely additive eigenvectors $E_{1}^{pa}(T^*_{\sigma})$.

Let $\mu \in E_{1}^{ca}(T^*_{\sigma})$. Clearly, $\mu \in C_{\sigma}$. Let $\mu \in E_{1}^{pa}(T^*_{\sigma})$. Then one can easily verify that $\mu$ satisfies Equation (\ref{cite}). From Lemma \ref{lemma3}, $E_{1}(T^*_{\sigma}) = E_{1}^{ca}(T^*_{\sigma}) \oplus E_{1}^{pa}(T^*_{\sigma}) \subseteq C_{\sigma} \oplus pa(\mathbb{N})$. 
\end{proof}

It is evident that the set formed by the intersection of eigenvectors of $T^*_{\sigma}$ for all $\sigma$ is a subset of $pa(\mathbb{N})$. Consequently, the proposition is a direct consequence of Theorem \ref{theorem2.3}.

\subsection{Proof of Proposition \ref{corollary9}}\label{proof-coro9}

Consider a binary relation $\succsim$ on $l_{\infty}$ that satisfies Axioms \ref{axiom1}-\ref{axiom5}, Axiom \ref{axiom6}, and IPIS. Now, let us assume, for the sake of contradiction, that there exists a constant equivalent function $I:l_{\infty} \to \mathbb{R}$ representing $\succsim$.

For each permutation $\sigma$, we define $F_{\sigma} = {p \in \Delta: \text{$p$ is an eigenvector of $T^*_{\sigma}$}}$. Clearly, $\lambda=1$ is the unique eigenvalue of $T^*_{\sigma}$. Consequently, we can establish that $$F_{\sigma} = \{p\in \Delta: \langle x,p\rangle = \langle x_{\sigma},p\rangle \text{ for every $x\in l_{\infty}$}\}.$$ 
It then follows from Theorem \ref{theorem2.3} that 
\begin{align*}
I(x) = \min_{p\in D} \langle x,p\rangle
\end{align*}
for some non-empty closed set $D \subseteq \bigcap_{\sigma} F_{\sigma}$.

For any $p\in D$, it follows that $\langle x,p\rangle = \langle x_{\sigma},p \rangle$ holds true for all $x\in l_{\infty}$ and for every permutation $\sigma$. Now, consider two infinite subsets, $A$ and $B$, of $\mathbb{N}$, with no common elements. There exists a bijection $f$ from $A$ to $B$. We can define a permutation $\sigma:\mathbb{N} \to \mathbb{N}$ as follows: $\sigma(i) = f(i)$ for all $i\in A$, $\sigma(i) = f^{-1}(i)$ for all $i\in B$, and $\sigma(i) = i$ for all other elements. Clearly, $\sigma$ is a permutation. We use the notation $\mathbb{I}_A$ to represent the sequence defined as $\mathbb{I}_A=xAy$, where $x=1$ and $y=0$. Thus, we can conclude that $\langle \mathbb{I}_A,p\rangle = \langle \mathbb{I}_{B},p \rangle$. This result demonstrates that any two infinite sets with no shared elements should have the same probability.

Now, let $A \subseteq \mathbb{N}$ such that both $A$ and its complement $A^c$ are infinite. We have $\langle \mathbb{I}_A,p\rangle = \langle \mathbb{I}_{A^c},p \rangle = \frac{1}{2}$. Now, given that $A$ is infinite, we can find a subset $B\subseteq A$ such that both $B$ and $A\setminus B$ are also infinite. This leads to the conclusion that $\langle \mathbb{I}_B,p\rangle = \langle \mathbb{I}_{A\setminus B},p \rangle = \frac{1}{4}$. Moreover, since $B^c$ is infinite, $\langle \mathbb{I}_B,p\rangle = \langle \mathbb{I}_{B^c},p \rangle = \frac{1}{2}$, which is a contradiction. Therefore, there does not exist any constant equivalent that represents a binary relation on $l_{\infty}$ satisfying Axioms \ref{axiom1}-\ref{axiom5}, Axiom \ref{axiom6}, and IPIS.

\newpage

\bibliographystyle{aea}

\bibliography{main}

\end{document}